\documentclass[12pt]{article}
\usepackage{amssymb,amsmath,amsthm,amsxtra,amsfonts}
\usepackage{geometry}
\usepackage{aliascnt}
\usepackage{color}
\usepackage[usenames,dvipsnames,table]{xcolor}
\usepackage{subcaption}
\usepackage{bbm}
\usepackage{dsfont}
\usepackage{lscape}
\usepackage{graphicx}
\usepackage{mathrsfs}
\usepackage{tikz}
\usetikzlibrary{plotmarks,arrows,calc,patterns}
\usepackage[round]{natbib}
\usepackage[colorlinks=true,citecolor=Blue,urlcolor=Blue,linkcolor=Blue]{hyperref}
\usepackage{algorithm,algpseudocode}
\usepackage{threeparttable}
\usepackage{booktabs}
\usepackage{rotating}
\usepackage{algorithm}

\newtheorem{theorem}{Theorem}[section]
\newtheorem{proposition}{Proposition}[section]

\newtheorem{corollary}{Corollary}[section]

\newtheorem{assumption}{Assumption}[section]
\newtheorem*{assumption*}{Assumption}

\newcommand{\R}{\ensuremath{\mathbf{R}}}
\newcommand{\Z}{\ensuremath{\mathbf{Z}}}

\newcommand{\Nn}{\ensuremath{\mathbf{N}}}

\newcommand{\E}{\ensuremath{\mathbb{E}}}

\newcommand{\dx}{\ensuremath{\mathrm{d}}}
\newcommand{\Var}{\ensuremath{\mathrm{Var}}}
\newcommand{\Cov}{\ensuremath{\mathrm{Cov}}}

\newcommand{\si}{\perp \! \! \! \perp}

\newcommand{\one}{\ensuremath{\mathds{1}}}

\usepackage{longtable}
\usepackage{xr}
\begin{document}
	\title{Machine Learning Panel Data Regressions with Heavy-tailed Dependent Data: Theory and Application}
	
	\author{Andrii Babii\thanks{University of North Carolina at Chapel Hill - Gardner Hall, CB 3305 Chapel Hill, NC 27599-3305. Email: babii.andrii@gmail.com.} \and Ryan T. Ball\thanks{Stephen M. Ross School of Business, University of Michigan, 701 Tappan Street, Ann Arbor, MI 48109. Email: rtball@umich.edu.} \and Eric Ghysels\thanks{Department of Economics and Kenan-Flagler Business School, University of North Carolina--Chapel Hill. Email: eghysels@unc.edu.} \and Jonas Striaukas\thanks{LIDAM UC Louvain and FRS\--FNRS Research Fellow. Email: jonas.striaukas@gmail.com.}}
	
	\maketitle
	
	\begin{abstract}
		\noindent {\footnotesize The paper introduces structured machine learning regressions for heavy-tailed dependent panel data potentially sampled at different frequencies. We focus on the sparse-group LASSO regularization. This type of regularization can take advantage of the mixed frequency time series panel data structures and improve the quality of the estimates. We obtain oracle inequalities for the pooled and fixed effects sparse-group LASSO panel data estimators recognizing that financial and economic data can have fat tails. To that end, we leverage on a new Fuk-Nagaev concentration inequality for panel data consisting of heavy-tailed $\tau$-mixing processes.}
	\end{abstract}
	
	\noindent%
	\vfill
\thispagestyle{empty}

\setcounter{page}{0}

\newpage

	\section{Introduction}
	
We analyze panel data regressions in a high-dimensional setting where the number of time-varying covariates can be very large and potentially exceed the sample size. We leverage on the structured sparsity approach using sparse-group LASSO (sg-LASSO) regularization for time series data with dictionaries. The advantages of this approach for individual time series data, potentially sampled at mixed frequencies, have been recently reported in \cite{babii2020machine}, who focus on nowcasting the US GDP growth in a data-rich environment. In this paper, we first show how to leverage on the sparse group regularization in a panel data setting. Second, we study the benefits of using the cross-sectional dimension for prediction with panel data paying particular attention to the issues of fat-tailed series which are relevant for the application involving financial time series. Third, we develop the debiased heteroskedasticity autocorrelation consistent (HAC) inference for regularized panel data regressions. Lastly, we provide an illustrative empirical example involving systematically predictable errors in analysts with individual firm earnings forecasts.
	
\smallskip 
	
Our paper relates to the literature on high-dimensional panel data models and the (group) LASSO regularization; see \cite{harding2019panel}, \cite{chiang2019post}, \cite{chernozhukov2019demand}, \cite{belloni2019high}, \cite{belloni2016inference}, \cite{lu2016shrinkage}, \cite{kock2016oracle}, \cite{su2016identifying}, \cite{farrell2015robust}, \cite{kock2013oracle}, \cite{lamarche2010robust}, \cite{koenker2004quantile}, among others. However, to the best of our knowledge, the existing literature relates mostly to the microeconometric problems and does not address comprehensively (1) the advantages of long panels; (2) the performance of regularized panel data estimators with potentially heavy-tailed covariates and regression errors, (3) the debiased HAC inference for regularized panel data, and (4) the sg-LASSO regularization of \cite{simon2013sparse} in a panel data setting.
	
\smallskip 
	
We recognize that the economic and financial time series data are often persistent with fat tails. To that end, we introduce a new Fuk-Nagaev concentration inequality for long panels. Using this inequality, we obtain oracle inequalities for the sg-LASSO that shed new light on how the predictive performance of pooled and fixed effect estimators scales with $N$ (cross-section) and $T$ (time series), which is especially relevant for modern panel data applications, where both $N$ and $T$ can be large; see \cite{fernandez2016individual}, \cite{hansen2007asymptotic}, \cite{alvarez2003time}, \cite{hahn2002asymptotically}, and \cite{phillips1999linear}, among others. Importantly, our theory covers the LASSO and the group-LASSO estimators as special cases of sg-LASSO.

\smallskip

In our empirical application we revisit a topic raised by \cite{ball2018automated} and \cite{carabias2018real}, but not resolved via formal inference in a high-dimensional setting. Namely, their empirical findings suggest that analysts tend to focus on their firm/industry when making earnings predictions while not fully taking into account the macroeconomic events affecting their firm/industry. More broadly, \cite{ball2018automated} argue that  analysts do not fully exploit information embedded in high-dimensional data and therefore {\it leave money on the table}. Thanks to the theoretical contributions in the current paper we can formally test that hypothesis in a data-rich environment. Note that, as \cite{ball2018automated} point out, it is important to take into account the mixed frequency nature of the data flow, which is why the machine learning panel regression methods presented in the paper apply to mixed frequency data. We use 26 predictors, including traditional macro and financial series as well as non-standard series generated by textual analysis of financial news. Using such a rich set of covariates, we test whether analyst' consensus earnings prediction errors are systematically related to either one of the aforementioned variables. 

\smallskip

The paper is organized as follows. Section \ref{sec:method} introduces the models and estimators. Oracle inequalities for sg-LASSO panel data regressions appear in Section \ref{sec:theory}. Section \ref{sec:inference} develops the debiased HAC inference for regularized panel data regressions. Monte Carlo simulations are reported in Section \ref{sec:mc}. The results of our empirical application are reported in Section \ref{sec:emp}. Section \ref{sec:conclusion} concludes. All technical details and detailed data descriptions appear in the Appendix and the Online Appendix.

\paragraph{Notation:} For a random variable $X\in\R$, let $\|X\|_q=(\E|X|^q)^{1/q}$ be its $L_q$ norm with $q\geq 1$. For $p\in\Nn$, put $[p] = \{1,2,\dots,p\}$. For a vector $\Delta\in\R^p$ and a subset $J\subset [p]$, let $\Delta_J$ be a vector in $\R^p$ with the same coordinates as $\Delta$ on $J$ and zero coordinates on $J^c$. Let $\mathcal{G}$ be a partition of $[p]$ defining the group structure, which is assumed to be known to the econometrician. For a vector $\beta\in\R^p$, the sparse-group structure is described by a pair $(S_0,\mathcal{G}_0)$, where $S_0=\{j\in[p]:\;\beta_j\ne 0 \}$ and  $\mathcal{G}_0 = \left\{G\in\mathcal{G}:\; \beta_{G} \ne 0\right\}$ are the support and respectively the group support of $\beta$. 

\smallskip

We also use $|S|$ to denote the cardinality of a set $S$. For $b\in\R^p$, its $\ell_q$ norm is denoted as $|b|_q$ = $(\sum_{j\in[p]}|b_j|^q)^{1/q}$ if $q\in[1,\infty)$ and $|b|_\infty = \max_{j\in[p]}|b_j|$ if $q=\infty$. For a  group structure $\mathcal{G}$, the $\ell_{2,1}$ group norm of $b\in\R^p$ is defined as $\|b\|_{2,1}=\sum_{G\in\mathcal{G}}|b_G|_2$. For $\mathbf{u},\mathbf{v}\in\R^J$, the empirical inner product is defined as $\langle \mathbf{u},\mathbf{v}\rangle_J = J^{-1}\sum_{j=1}^J u_jv_j$ with the induced empirical norm $\|.\|_J^2=\langle.,.\rangle_J=|.|_2^2/J$. For a symmetric $p\times p$ matrix $A$, let $\mathrm{vech}(A)\in\R^{p(p+1)/2}$ be its vectorization consisting of the lower triangular and the diagonal elements. Let $A_G$ be a sub-matrix consisting of rows of $A$ corresponding to indices in $G\subset[p]$. If $G=\{j\}$ for some $j\in[p]$, then we simply write $A_G=A_j$. Let $\|A\|_\infty=\max_{j\in[p]}|A_j|$ be the matrix norm. For $a,b\in\R$, we put $a\vee b = \max\{a,b\}$ and $a\wedge b = \min\{a,b\}$. Lastly, we write $a_n\lesssim b_n$ if there exists a (sufficiently large) absolute constant $C$ such that $a_n\leq C b_n$ for all $n\geq 1$ and $a_n\sim b_n$ if $a_n\lesssim b_n$ and $b_n\lesssim a_n$.

\section{High-dimensional (mixed frequency) panels}\label{sec:method}
Motivated by our empirical application, we allow the high-dimensional set of predictors to be sampled at a higher frequency than the target variable. Let $K$ be the total number of time-varying predictors $\{x_{i,t-(j-1)/m,k}:i\in[N],t\in[T],j\in[m],k\in[K]\}$ possibly measured at some higher frequency with $m$ observations for every low-frequency period $t\in[T]$ and every entity $i\in[N]$. Consider the following (mixed frequency) panel data regression
\begin{equation*}
	y_{i,t+h} = \alpha_i + \sum_{k=1}^{K}\psi(L^{1/m};\beta_k)x_{i,t,k} + u_{i,t},
\end{equation*}
where $h\geq 0$ is the prediction horizon, $\alpha_i$ is the entity-specific intercept, and
\begin{equation}\label{eq:hf_lag}
	\psi(L^{1/m};\beta_k)x_{i,t,k} = \frac{1}{m}\sum_{j=1}^{m}\beta_{j,k}x_{i,t-(j-1)/m,k}
\end{equation}
is a high-frequency lag polynomial with $\beta_k=(\beta_{1,k},\dots,\beta_{m,k})^\top\in\R^m$. More generally, the frequency can also be specific to the predictor $k\in[K]$, in which case we would have $m_k$ instead of $m$. We can also absorb the (low-frequency) lags of $y_{i,t}$ in covariates. When $m$ = 1, we retain the standard panel data regression model
\begin{equation*}
	y_{i,t+h} = \alpha_i + \sum_{k=1}^K\beta_{k}x_{i,t,k} + u_{i,t},
\end{equation*}
while $m>1$ signifies that the high-frequency lags of $x_{i,t,k}$ are also included. The large number of predictors $K$ with potentially large number of high-frequency measurements $m$ can be a rich source of predictive information, yet at the same time, estimating $N + m\times K$ parameters is  costly and may reduce the predictive performance in small samples.

\smallskip

To reduce the proliferation of lag parameters, we follow the MIDAS literature; see \cite{ghysels2006predicting}, \cite{ghysels:mrf}, and \cite{babii2020inference,babii2020machine}. Instead of estimating $m$ individual slopes of high-frequency covariate $k\in[K]$ in equation (\ref{eq:hf_lag}), with some abuse of notation, we estimate a weight function $\omega$ parameterized by $\beta_k\in\R^L$ with $L< m$
\begin{equation*}
	\psi(L^{1/m};\beta_k)x_{i,t,k} = \frac{1}{m}\sum_{j=1}^m\omega\left(\frac{j-1}{m};\beta_k\right)x_{i,t-(j-1)/m,k},
\end{equation*}
where
\begin{equation*}\label{eq:midas_weight}
	\omega(s;\beta_k) = \sum_{l=0}^{L-1}\beta_{l,k}w_l(s),\qquad \forall s\in[0,1]
\end{equation*}
and $(w_l)_{l\geq 0}$ is a collection of $L$ approximating functions, called the \textit{dictionary}. An example of a dictionary is the set of orthogonal Legendre polynomials on $[0,1]$ that can be computed via the Rodrigues' formula $w_l(s)=\frac{1}{l!}\frac{\dx^l}{\dx s^l}(s^2-s)^l.$\footnote{The Legendre polynomials have the universal approximation property and can approximate any continuous function uniformly on $[0,1]$. At the same time they can generate a rich family of MIDAS weights with a relatively small number of parameters which is attractive in time series applications where the signal-to-noise ratio is often low.} For instance, the first five elements are
\begin{equation*}
\begin{aligned}
	w_0(s) & = 1 \\
	w_1(s) & = 2s-1 \\
	w_2(s) & = 6s^2-6s + 1 \\
	w_3(s) & = 20s^3 - 30s^2 + 12s - 1 \\
	w_4(s) & = 70s^4 - 140s^3 + 90s^2 - 20s + 1.
\end{aligned}
\end{equation*}
More generally, we can use Gegenbauer polynomials, trigonometric polynomials, or wavelets. The orthogonal polynomials usually have better numerical properties than their popular non-orthogonal counterpart, such as the \cite{almon1965distributed} lag structure. The attractive feature of linear in parameters dictionaries is that we can map the MIDAS regression to the linear regression framework that can be solved via a convex optimization. To that end, define $\mathbf{x}_i = (X_{i,1}W,\dots,X_{i,K}W)$, where for each $k\in[K]$, $X_{i,k} = (x_{i,t-(j-1)/m,k})_{t\in[T],j\in[m]}$ is a $T\times m$ matrix of predictors and $W=(w_l((j-1)/m)/m)_{j\in[m],0\leq l\leq L-1}$ is an $m\times L$ matrix corresponding to the dictionary $(w_l)_{l\geq 0}$. In addition, let 
$\mathbf{y}_i$ = $(y_{i,1+h},\dots,y_{i,T+h})^\top$ and $\mathbf{u}_i$ = $(u_{i,1},\dots,u_{i,T})^\top.$
Then the regression equation after stacking time series observations for each $i\in[N]$ is
\begin{equation*}
	\mathbf{y}_i = \iota\alpha_i + \mathbf{x}_i\beta + \mathbf{u}_i,
\end{equation*}
where $\iota\in\R^T$ is the all-ones vector and $\beta\in\R^{LK}$ is a vector of slopes. Lastly, put $\mathbf{y} = (\mathbf{y}_1^\top,\dots, \mathbf{y}_N^\top)^\top$, $\mathbf{X}=(\mathbf{x}_1^\top, \dots, \mathbf{x}_N^\top)^\top$, and $\mathbf{u} = (\mathbf{u}_1^\top,\dots,\mathbf{u}_N^\top)^\top$. Then the regression equation after stacking all cross-sectional observations is
\begin{equation*}
\mathbf{y} = B\alpha + \mathbf{X}\beta + \mathbf{u},
\end{equation*}
where $B=I_N\otimes\iota$, $\alpha=(\alpha_1,\dots,\alpha_N)$, and $\otimes$ is the Kronecker product.

\smallskip

The MIDAS approach allows us to effectively reduce the dimensionality pertaining to the high-frequency lags. Alternatively, we may apply what is known as the UMIDAS scheme, see e.g., \cite{foroni2015unrestricted}, and directly estimate the coefficients associated with each high-frequency covariate lags separately (see equation (\ref{eq:UMIDAS}) in Section \ref{sec:mc} for example). Such a strategy, which as \cite{foroni2015unrestricted} argue works in single regressions when the ratio high to low-frequency sampling is small, may not be appealing in high-dimensional cases, as the estimation and prediction performance deteriorates due to the potentially large number of coefficients; see \cite{babii2020machine} for further discussion. Also, while assuming that the individual lag coefficients in equation~(\ref{eq:hf_lag}) are approximately sparse is \textit{highly} restrictive, the approximate sparsity of slopes of the dictionary elements $(w_l)_{l\geq 0}$ is plausible. For instance, if $w_0(s)=1$ with $\beta_{0,k}\ne 0$ and $\beta_{l,k}=0,\forall l\geq 1$, we recover the averaging of high-frequency lags of covariate $k$ as a special case. More generally, the weight $\omega$ may be a decreasing function over lags and we may want to learn its shape from the data maximizing the predictive performance.\footnote{See \cite{ball2013dissecting} and \cite{ball2018mixed} for further discussion on interpreting the shape of MIDAS polynomials in accounting data applications considered in our empirical application.}

\smallskip

Given that the number of potential predictors $K$ can be large, additional regularization can improve the predictive performance in small samples. To that end, we take advantage of the sg-LASSO regularization that was shown to be attractive for individual time series ML regressions in \cite{babii2020machine}. The fixed effects panel data estimator with sparse-group regularization solves
\begin{equation}\label{eq:sgl}
	\min_{(a,b)\in\R^{N+LK}}\|\mathbf{y} - Ba - \mathbf{X}b\|_{NT}^2 + 2\lambda\Omega(b),
\end{equation}
where $\|.\|_{NT}^2 = |.|^2/(NT)$ is the empirical norm and
\begin{equation*}
	\Omega(b) = \gamma|b|_1 + (1-\gamma)\|b\|_{2,1}
\end{equation*}
is a regularizing functional, which is a linear combination of LASSO and group LASSO penalties. The parameter $\gamma\in[0,1]$ determines the relative weights of the $\ell_1$ (sparsity) and the $\ell_{2,1}$ (group sparsity) norms, while the amount of regularization is controlled by the regularization parameter $\lambda\geq 0$. Recall also that for a group structure $\mathcal{G}$ described as a partition of $[p]=\{1,2,\dots,p\}$, the group LASSO norm is computed as $\|b\|_{2,1}=\sum_{G\in\mathcal{G}}|b_G|_2$. The group structure is assumed to be known to the econometrician, which in our setting corresponds to time series lags of covariates. More generally, we may also combine covariates of a similar nature in groups. Throughout the paper we assume that groups have fixed size, which is well-justified in our empirical applications.\footnote{See \cite{babii2020high} for a continuous-time mixed-frequency regression where the group size is allowed to increase with the sample size under the in-fill asymptotics.} Therefore, the selection of covariates is performed by the group LASSO penalty, which encourages sparsity between groups. In addition, the $\ell_1$ LASSO norm  promotes sparsity within groups and allows us to learn the shape of the MIDAS weights from the data. 

\smallskip

It is worth mentioning that the linear in parameters approximation to the MIDAS weight function leads to the convex optimization parameter problem  in equation~(\ref{eq:sgl}) that can be solved efficiently, e.g., via the proximal gradient descent algorithm, or its block-coordinate descent versions. In contrast, a popular beta weights leads to a nonlinear non-convex optimization problem that becomes challenging to solve in high-dimensions; cf.\ \cite{marsilli2014variable} and \cite{khalaf2020dynamic}.

\section{Oracle inequalities}\label{sec:theory}
In this section, we provide the theoretical analysis of predictive performance of regularized panel data regressions with the sg-LASSO regularization, including the standard LASSO and the group LASSO regularizations as special cases. It is worth stressing that the analysis of this section is not tied to the mixed-frequency data setting and applies to the generic high-dimensional panel data regularized with the sg-LASSO penalty function. Importantly, we focus on panels consisting of potentially persistent $\tau$-mixing time series with polynomial tails. Consider a generic panel data projection with a countable number of predictors
\begin{equation*}
	y_{i,t+h} = \alpha_i + \sum_{j=1}^\infty \beta_jx_{i,t,j} + u_{i,t},\qquad \E[u_{i,t}x_{i,t,j}]=0,\quad \forall j\geq 1,
\end{equation*}
This model subsumes the mixed-frequency data regressions as a special case, in which case covariates are obtained, e.g., from the aggregation with Legendre polynomials. The covariates may also include the time-varying covariates common for all entities (macroeconomic factors), lags of $y_{i,t}$, the intercept, as well as additional lags of a baseline covariate. 

\subsection{$\tau$-mixing}
We measure the persistence of the data with $\tau$-mixing coefficients. For a $\sigma$-algebra $\mathcal{M}$ and a random vector $\xi\in\R^l$, put
\begin{equation*}
	\tau(\mathcal{M},\xi) = \left\|\sup_{f\in\mathrm{Lip}_1}\left|\E(f(\xi)|\mathcal{M}) - \E(f(\xi))\right|\right\|_1,
\end{equation*}
where $\mathrm{Lip}_1=\{f:\R^l\to\R:\; |f(x) - f(y)| \leq |x-y|_1 \}$ is a set of $1$-Lipschitz functions from $\R^l$ to $\R$.\footnote{See \cite{dedecker2004coupling} and \cite{dedecker2005new} for equivalent definitions.} For a stochastic process $(\xi_t)_{t\in\Z}$ with a natural filtration generated by its past $\mathcal{M}_t=\sigma(\xi_t,\xi_{t-1},\dots)$, the $\tau$-mixing coefficients are defined as
\begin{equation*}
\tau_k = \sup_{j\geq 1}\frac{1}{j}\sup_{t+k\leq t_1<\dots<t_j}\tau(\mathcal{M}_t,(\xi_{t_1},\dots,\xi_{t_j})),\qquad k\geq 0
\end{equation*}
where the supremum is taken over all $t,t_1,\dots,t_j\in\Z$. If $\tau_k\downarrow0$, as $k\uparrow\infty$ then the process is called $\tau$-mixing. The class of $\tau$-mixing processes can be placed somewhere between the $\alpha$-mixing processes and mixingales --- the $\tau$-mixing condition is less restrictive than the $\alpha$-mixing condition,\footnote{The class of $\alpha$-mixing processes is too restrictive for the predictive linear projection model with covariates and autoregressive lags; see also \cite{babii2020machine}, Proposition A.3.1.} yet at the same time, there exists a convenient for us coupling result for $\tau$-mixing processes, which is not the case for the mixingales or near-epoch dependent processes; see \cite{dedecker2003new} and \cite{dedecker2004coupling,dedecker2005new} for more details.  This allows us to obtain concentration inequalities and performance guarantees for the sg-LASSO estimator; see Appendix~\ref{sec:concentration} for more details.

\subsection{Pooled regression}
For pooled regressions, we assume that all entities share the same intercept parameter $\alpha_1=\dots=\alpha_N=\alpha$. The pooled sg-LASSO estimator $\hat\rho=(\hat\alpha,\hat\beta^\top)^\top$ solves
\begin{equation}\label{eq:pooled_panel}
	\min_{r=(a,b)\in\R^{1+p}}\|\mathbf{y} - a\iota - \mathbf{X}b\|_{NT}^2 + 2 \lambda\Omega(r).
\end{equation}

\smallskip	

Define (a) $z_{i,t}=(1,x_{i,t}^\top)^\top$, where $x_{i,t}\in\R^p$ is a vector of predictors, (b) $u_i=(u_{i,1},\dots,u_{i,T})$ and (c) $x_i = (x_{i,1}^\top,\dots,x_{i,T}^\top)^\top$ for $i\in[N]$. The following assumption imposes mild restrictions on the data.

\begin{assumption}[Data]\label{as:data}
	$\{(u_i,x_i^\top)^\top: i\in\Nn\}$ are independent vectors in $\R^{(p+1)}\times\R^T$ such that (i) $\max_{i\in [N],t\in[T],j\in[p+1]}\|u_{i,t}z_{i,t,j}\|_q = O(1)$ for some $q>2$; (ii) the $\tau$-mixing coefficients of $(u_{i,t}z_{i,t})_{t\in\Z}$ satisfy $\max_{i\in[N],j\in[p+1]}\tau_{k-1}^{(i,j)}=O(k^{-a}),\forall k\geq 1$ with  $a>(q-1)/(q-2)$; (iii) $\max_{i\in[N],t\in[T],j,k\in[p+1]}\|z_{i,t,j}z_{i,t,k}\|_{\tilde q} = O(1)$ for some $\tilde q>2$; (iv) the $\tau$-mixing coefficients of $\mathrm{vech}((z_{i,t}z_{i,t}^\top))_{t\in\Z}$ satisfy $\max_{i\in[N],j\in[(p+1)(p+2)/2]}\tilde \tau_{k-1}^{(i,j)}\leq \tilde ck^{-\tilde a},\forall k\geq 1$ with $\tilde c>0$ and $\tilde a>(\tilde q - 1)/(\tilde q - 2)$.
\end{assumption}
Note that we do not impose stationarity over $t\in\Z$ and require that only $2+\epsilon$ moments exist with $\epsilon>0$, which is a realistic assumption in our empirical application and more generally for datasets encountered in time series and financial econometrics applications. Note also that the time series dependence is assumed to fade away relatively slowly --- at a polynomial rate as measured by the $\tau$-mixing coefficients.

\smallskip

Next, we assume that the $(1+p)\times(1+p)$ matrix $\Sigma_{N,T}=\frac{1}{NT}\sum_{i=1}^N\sum_{t=1}^T\E[z_{i,t}z_{i,t}^\top]$ exists and is non-singular uniformly over $N,T,p$:
\begin{assumption}[Covariance matrix]\label{as:covariance}
	The smallest eigenvalue of $\Sigma_{N,T}$ is uniformly bounded away from zero by some universal constant $\gamma_{\min}>0$.
\end{assumption}
Assumption~\ref{as:covariance} is satisfied for the spiked identity and Topelitz covariance structures. It can be interpreted as a completeness condition, see \cite{babii2017completeness}, and can also be relaxed to the restricted eigenvalue condition imposed on the population covariance matrix $\Sigma_{N,T}$; see \cite{babii2020machine}. We can also allow for $\gamma_{\min}\downarrow0$ as $N,T,p\uparrow\infty$, in which case $\gamma_{\min}^{-1}$ would slow down the convergence rates in oracle inequalities and could be interpreted as a measure of ill-posedness; see also \cite{carrasco2007linear}.

\smallskip

Lastly, we assume that the regularization parameter $\lambda$ scales appropriately with the number of covariates $p$, the length of the panel $T$, the size of the cross-section $N$, and a certain exponent $\kappa$ that depends on the tail parameter $q$ and the persistence parameter $a$. The precise order of the regularization parameter is described by the Fuk-Nagaev inequality for long panels appearing in the Appendix; see Theorem~\ref{thm:fn_long}.
\begin{assumption}[Regularization]\label{as:tuning}
	For some $\delta\in(0,1)$
	\begin{equation*}
		\lambda \sim \left(\frac{p}{\delta (NT)^{\kappa - 1}}\right)^{1/\kappa}\vee \sqrt{\frac{\log(p/\delta)}{NT}},
	\end{equation*}
	where $\kappa= ((a+1)q - 1) / (a+q-1)$ and $a,q$ are as in Assumptions~\ref{as:data}.
\end{assumption}
Our first result is the oracle inequality for the pooled sg-LASSO estimator described in equation~(\ref{eq:pooled_panel}). The result allows for misspecified regressions with a non-trivial approximation error in the sense that we consider more generally
\begin{equation*}
	\mathbf{y} = \mathbf{m} + \mathbf{u},
\end{equation*}
where $\mathbf{m}\in\R^{NT}$ is approximated with $\mathbf{Z}\rho$, $\mathbf{Z}=(\iota,\mathbf{X})$, $\iota\in\R^{NT}$ is all-ones vector, and $\rho = (\alpha,\beta^\top)^\top$. The approximation error $\mathbf{m}-\mathbf{Z}\rho$ might come from the fact that the MIDAS weight function may not have the exact expansion in terms of the specified dictionary or from the fact that some of the relevant predictors are not included in the regression equation. To state the result, let $S_0=\{j\in[p]:\; \beta_j\ne 0 \}$ be the support of $\beta$ and let $\mathcal{G}_0=\{G\in\mathcal{G}:\;\beta_G\ne 0 \}$ be the group support of $\beta$. Consider the \textit{effective sparsity} of the sparse-group structure, defined as  $s^{1/2} = \gamma\sqrt{{|S_0|}} + (1-\gamma)\sqrt{{|\mathcal{G}_0|}}$. Note that $s$ is proportional to the sparsity $|S_0|$, when $\gamma=1$ and to the group sparsity $|\mathcal{G}_0|$ when $\gamma = 0$. Define $r_{N,T}^{\rm pooled} = s^{\tilde\kappa}p^2/(NT)^{\tilde \kappa - 1} + p^2\exp(-cNT/s^2)$.

\begin{theorem}\label{thm:pooled_panel}
	Suppose that Assumptions~\ref{as:data}, \ref{as:covariance}, and \ref{as:tuning} are satisfied. Then with probability at least $1 - \delta - O(r_{N,T}^{\rm pooled})$
	\begin{equation*}
		\|\mathbf{Z}(\hat\rho - \rho)\|^2_{NT} \lesssim s\lambda^2 + \|\mathbf{m} - \mathbf{Z}\rho\|_{NT}^2
	\end{equation*}
	and
	\begin{equation*}
		|\hat\rho - \rho|_1 \lesssim s\lambda + \lambda^{-1}\|\mathbf{m} - \mathbf{Z}\rho\|_{NT}^2 + s^{1/2}\|\mathbf{m} - \mathbf{Z}\rho\|_{NT},
	\end{equation*}
	for some $c>0$ and $\tilde\kappa = ((\tilde a + 1)\tilde q - 1)/(\tilde a + \tilde q - 1)$.
\end{theorem}
The proof of this result can be found in the Appendix. Theorem~\ref{thm:pooled_panel} describes the non-asymptotic oracle inequalities for the prediction and the estimation accuracy in the environment where the number of regressors $p$ is allowed to scale with the effective sample size $NT$. Importantly, the result is stated under the weak tail and persistence conditions in Assumption~\ref{as:data}. Parameters $\kappa$ and $\tilde\kappa$ are the dependence-tails exponents for stochastic processes driving the regression score and the covariance matrix respectively. Theorem~\ref{thm:pooled_panel} shows that the prediction and the estimation accuracy of pooled panel data regressions improves when the sparse-group structure is taken into account. Indeed, for the LASSO regression, the effective sparsity reduces to $s^{1/2}=\sqrt{|S_0|}$, which is larger than $\gamma\sqrt{|S_0|} + (1-\gamma)\sqrt{|\mathcal{G}_0|}$ in the case of sg-LASSO.

\smallskip
	
Next, we consider the convergence rates of the prediction and estimation errors. The following assumption considers a simplified setting, where the approximation error vanishes sufficiently fast, and the total number of regressors vanishes sufficiently fast with the effective sample size $NT$.

\begin{assumption}\label{as:rates}
	(i) $\|\mathbf{m} - \mathbf{Z}\rho\|^2_{NT} = O_P(s\lambda^2)$; and (ii) $s^{\tilde\kappa}p^2(NT)^{1-\tilde \kappa}\to0$ and $p^2\exp(-cNT/s^2)\to 0$.
\end{assumption}
Note that Assumption~\ref{as:rates} allows for (1) $N\to\infty$ while $T$ is fixed; (2) $T\to\infty$ while $N$ is fixed; and (3) both $N\to\infty$ and $T\to\infty$ without restricting the relative growth of the two. The following result describes the prediction and the estimation convergence rates in the asymptotic environment outlined in Assumption~\ref{as:rates} and is an immediate consequence of Theorem~\ref{thm:pooled_panel}.
\begin{corollary}\label{cor:pooled}
	Suppose that Assumptions~\ref{as:data}, \ref{as:covariance}, \ref{as:tuning}, and \ref{as:rates} are satisfied. Then
	\begin{equation*}
		\|\mathbf{Z}(\hat\rho - \rho)\|^2_{NT} = O_P\left(\frac{sp^{2/\kappa}}{(NT)^{2-2/\kappa}}\vee \frac{s\log p}{NT}\right)
	\end{equation*}
	and
	\begin{equation*}
		|\hat\rho - \rho|_1 = O_P\left(\frac{sp^{1/\kappa}}{(NT)^{1-1/\kappa}}\vee s\sqrt{\frac{\log p}{NT}}\right).
	\end{equation*}
\end{corollary}
Corollary~\ref{cor:pooled} describes the prediction and the estimation accuracy of pooled sparse-group panel data regressions. It suggests that the predictive performance of the sg-LASSO (and consequently LASSO and group LASSO) regressions may deteriorate when regression errors and/or predictors are heavy-tailed or when the data are extremely persistent. However, for geometrically ergodic Markov processes, e.g., stationary AR(1) process, the $\tau$-mixing coefficients decline geometrically fast, so that $\kappa\approx q$ and $\tilde\kappa\approx \tilde q$. In this case, the prediction accuracy scales approximately at the rate $O_P\left(\frac{p^{2/q}}{(NT)^{2-2/q}}\vee \frac{\log p}{NT}\right)$ and the predictive performance may be affected only by the tails constant $q$. 

\smallskip

If additionally the data are sub-Gaussian, then moments of all order $q\geq 2$ exist, and for any particular effective sample size $NT$, the first term can be made arbitrarily small relatively to the second term. In this case we recover the $O_P\left(\frac{\log p}{NT}\right)$ rate typically obtained for sub-Gaussian data. On the other hand, if the polynomial tail dominates, then we need $p = o((NT)^{q -1})$ for the prediction and the estimation consistency provided that $\tilde q\geq 2q-1$ and the sparsity constant $s$ is fixed. In this case, we have a \textit{significantly weaker} requirement than the $p=o(T^{q - 1})$ condition needed for time series regressions in \cite{babii2020machine}. Moreover, since $q>2$, $p=o((NT)^{q-1})$ can be significantly weaker than the $p=o(NT)$ condition typically needed for QMLE/GMM estimators without regularization.

\smallskip

Theorem~\ref{thm:pooled_panel} and Corollary~\ref{cor:pooled} imply two practical consequences: (1) one may want to exclude (or suitably transform) the heavy-tailed series from the high-dimensional predictive regressions based on the preliminary estimates of the tail index, e.g., using the Hill estimator; (2) if the individual heterogeneity can be ignored, then pooling panel data can improve significantly the predictive performance. In the latter case, one can also preliminary cluster similar series in groups, e.g., based on the unsupervised clustering algorithms, which may strike a good balance between the pooling benefits and heterogeneity.

\subsection{Fixed effects}
Pooled regressions are attractive since the effective sample size $NT$ can be huge, yet the heterogeneity of individual time series may be lost. If the underlying series have a substantial heterogeneity over $i\in[N]$, then taking this into account might reduce the projection error and improve the predictive accuracy. At a very extreme side, the cross-sectional structure can be completely ignored and individual time series regressions can be used for prediction. The fixed effects panel data regressions strike a good balance between the two extremes controlling for heterogeneity with entity-specific intercepts.

\smallskip

The fixed effects sg-LASSO estimator $\hat\rho=(\hat\alpha^\top,\hat\beta^\top)^\top$ solves
\begin{equation*}
	\min_{(a,b)\in\R^{N+p}}\|\mathbf{y} - Ba - \mathbf{X}b \|_{NT}^2 + 2 \lambda\Omega(b),
\end{equation*}
where $B=I_N\otimes\iota$, $I_N$ is $N\times N$ identity matrix, $\iota\in\R^T$ is an all-ones vector, and $\Omega$ is the sg-LASSO regularizing functional. It is worth stressing that the design matrix $\mathbf{X}$ does not include the intercept and that we do not penalize the fixed effects, that are typically not sparse. By Fermat's rule, the first-order conditions are
\begin{equation}\label{eq:focs}
	\begin{aligned}
		\hat\alpha & = (B^\top B)^{-1}B^\top(\mathbf{y} - \mathbf{X}\hat\beta) \\
		0 & = \mathbf{X}^\top M_B(\mathbf{X}\hat\beta - \mathbf{y})/NT + \lambda z^*
	\end{aligned}
\end{equation}
for some $z^*\in\partial\Omega(\hat\beta)$, where $b\mapsto \partial\Omega(b)$ is the subdifferential  of $\Omega$ and $M_B = I - B(B^\top B)^{-1}B^\top$ is the orthogonal projection matrix. It is easy to see from the first-order conditions that the estimator of $\hat\beta$ is equivalent to 1) penalized GLS estimator for the first-differenced regression; 2) penalized OLS estimator for the regression written in the deviation from time means; and 3) penalized OLS estimator where the fixed effects are partialled-out. Therefore, the equivalence between the three approaches is not affected by the penalization; cf.\ \cite{arellano2003panel} for low-dimensional panels.

\smallskip

With some abuse of notation, redefine
\begin{equation}\label{eq:covariance_fe}
	\hat\Sigma_{N,T} = \begin{pmatrix}
		\frac{1}{T}B^\top B & \frac{1}{\sqrt{N}T}B^\top\mathbf{X} \\
		\frac{1}{\sqrt{N}T}\mathbf{X}^\top B & \frac{1}{NT}\mathbf{X}^\top\mathbf{X}
	\end{pmatrix}\quad \text{and}\quad \Sigma_{N,T} = \begin{pmatrix}
		I_N & \frac{1}{\sqrt{N}T}\E\left[B^\top\mathbf{X}\right] \\
		\frac{1}{\sqrt{N}T}\E\left[\mathbf{X}^\top B\right] & \E[x_{i,t}x_{i,t}^\top]
	\end{pmatrix}.
\end{equation}
We will assume that the smallest eigenvalue of $\Sigma_{N,T}$ is uniformly bounded away from zero by some constant. Note that if $x_{i,t}\sim N(0,I_p)$, then $\Sigma_{N,T}=I_{N+p}$ and this assumption is trivially satisfied.

\smallskip

The order of the regularization parameter is governed by the Fuk-Nagaev inequality for long panels; see Appendix, Theorem~\ref{thm:fn_long}.
\begin{assumption}[Regularization]\label{as:tuning_fe}
	For some $\delta\in(0,1)$
	\begin{equation*}
		\lambda \sim \left(\frac{p\vee N^{\kappa/2}}{\delta (NT)^{\kappa -1}}\right)^{1/\kappa}\vee \sqrt{\frac{\log(p\vee N/\delta)}{NT}},
	\end{equation*}
	where $\kappa=((a+1)q - 1)/(a+q-1)$, and $a,q$ are as in Assumptions~\ref{as:data}.
\end{assumption}

Similarly to the pooled regressions, we state the oracle inequality allowing for the approximation error. For fixed effects regressions, with some abuse of notation we redefine $\mathbf{Z} = (B,\mathbf{X})$ and $\rho = (\alpha^\top,\beta^\top)^\top$. Put also $r_{N,T}^{\rm fe}=p(s\vee N)^{\tilde \kappa} T^{1-\tilde \kappa}(N^{1-\tilde \kappa/2} + pN^{1-\tilde\kappa}) + p(p\vee N)e^{-cNT/(s\vee N)^2}$ with $\tilde\kappa = ((\tilde a + 1)\tilde q - 1) / (\tilde a + \tilde q - 1)$ and some $c>0$.
\begin{theorem}\label{thm:fixed_effects}
	Suppose that Assumptions~\ref{as:data}, \ref{as:covariance}, and \ref{as:tuning_fe} are satisfied. Then with probability at least $1 - \delta - O(r_{N,T}^{\rm fe})$
	\begin{equation*}
		\|\mathbf{Z}(\hat\rho - \rho)\|^2_{NT} \lesssim (s\vee N)\lambda^2 + \|\mathbf{m} - \mathbf{Z}\rho\|_{NT}^2.
	\end{equation*}
\end{theorem}
Theorem~\ref{thm:fixed_effects} states a non-asymptotic oracle inequality for the prediction error in the fixed effects panel data regressions estimated with the sg-LASSO. To see clearly, how the prediction accuracy scales with the sample size, we make the following assumption.

\begin{assumption}\label{as:rates_fe}
	Suppose that (i) $\|\mathbf{m} - \mathbf{Z}\rho\|^2_{NT} = O_P((s\vee N)\lambda^2)$; (ii) $(p+N^{\tilde\kappa/2})p(s\vee N)^{\tilde\kappa}N^{1-\tilde\kappa}T^{1-\tilde\kappa} \to 0$ and $p(p\vee N)e^{-cNT/(s\vee N)^2} \to 0$.
\end{assumption}

The following corollary is an immediate consequence of Theorem~\ref{thm:fixed_effects}.
\begin{corollary}\label{cor:fixed_effects}
	Suppose that Assumptions~\ref{as:data}, \ref{as:covariance}, \ref{as:tuning_fe}, and \ref{as:rates_fe} are satisfied. Then
	\begin{equation*}
		\|\mathbf{Z}(\hat\rho - \rho)\|^2_{NT} = O_P\left(\frac{(s\vee N)(p^{2/\kappa}\vee N)}{N^{1-2/\kappa}T^{2 - 2/\kappa}}\vee {\frac{(s\vee N)\log(p\vee N)}{NT}}\right).
	\end{equation*}
\end{corollary}
Corollary~\ref{cor:fixed_effects} allows for $s,p,N,T\to\infty$ at appropriate rates. However, we pay an additional price for estimating $N$ fixed effects which plays a similar role to the effective dimension of covariates. An immediate practical implication is that to achieve accurate predictions with high-dimensional fixed effect regressions, the panel has to be sufficiently long to offset the estimation error of the individual fixed effects. Likewise, the tails and the persistence of the data may also reduce the prediction accuracy in small samples through $\kappa$, which is approximately equal to $q$ for geometrically decaying $\tau$-mixing coefficients.

\section{Debiased inference}\label{sec:inference}
In this section, we develop the debiased inferential methods for pooled panel data regressions.
For a vector $\rho\in\R^{p+1}$, we use $\rho_G\in\R^{|G|}$ to denote the subvector of elements of $\rho\in\R^{p+1}$ indexed by $G\subset[p+1]$. Let $B = \hat\Theta\mathbf{Z}^\top(\mathbf{y} - \mathbf{Z}\hat\rho) / NT$ denote the bias-correction for the sg-LASSO estimator, where $\hat\Theta$ is the nodewise LASSO estimator of the precision matrix $\Theta=\Sigma^{-1}$, where $\Sigma=\E[z_{i,t}z_{i,t}^\top]$. For pooled panel data, this estimator can be obtained as follows:
\begin{enumerate}
	\item For each $j\in[p+1]$, let $\hat\mu_j = (\hat\mu_{j,1},\dots,\hat\mu_{j,p})^\top$ be a solution to
	\begin{equation*}
		\min_{\mu\in\R^{p}}\|\mathbf{Z}_j - \mathbf{Z}_{-j}\mu\|^2_{NT} + 2\lambda_j|\mu|_1,
	\end{equation*}
	where $\mathbf{Z}_j$ is $NT\times 1$ vector of stacked observations $\{z_{i,t,j}\in\R:i\in[N],t\in[T]\}$ and $\mathbf{Z}_{-j}$ is the $NT\times p$ matrix of stacked observations $\{(z_{i,t,k})_{k\ne j}\in\R^{p}: i\in[N],t\in[T] \}$. Put
	\begin{equation*}
		\hat\sigma_j^2 = \|\mathbf{Z}_j - \mathbf{Z}_{-j}\hat \mu_j\|^2_{NT} + \lambda_j|\hat\mu_j|,
	\end{equation*}
	
	\item Compute  $\hat\Theta = \hat B^{-1}\hat C$, where $\hat B = \mathrm{diag}(\hat\sigma_1^2,\dots,\hat\sigma_{p+1}^2)$, and
	\begin{equation*}
	\hat C = \begin{pmatrix}
	1 & -\hat\mu_{1,1} &  \dots & -\hat\mu_{1,p} \\
	-\hat\mu_{2,1} & 1 & \dots & -\hat\mu_{2,p} \\
	\vdots & \vdots & \ddots  & \vdots \\
	-\hat\mu_{p,1} & \dots & -\hat\mu_{p,p}& 1
	\end{pmatrix}.
	\end{equation*}
\end{enumerate}
Let $v_{i,t,j}=z_{i,t,j}-\sum_{k\ne j}\mu_{j,k}z_{i,t,k}$ be the regression error for $j^{\rm th}$ nodewise LASSO regression. Let $s_j$ be the number of non-zero elements in $j^{\rm th}$ row of precision matrix $\Theta_j$, and put $S=\max_{j\in G}s_j$, and $s^*=s\vee S$.

\smallskip

The following assumption describes an additional set of conditions for the debiased central limit theorem.
\begin{assumption}\label{as:clt}
	(i) $\sup_z\E[u_{i,t}^2|z_{i,t}=z]=O(1)$; (ii) $\|\Theta_G\|_\infty = O(1)$ for $G\subset[p+1]$ of fixed size; (iii) the long run variance of $(u_{i,t}^2)_{t\in\Z}$ and $(v_{i,t,j}^2)_{t\in\Z}$ exists for every $j\in G$; (iv) $s^{*2}\log^2 p/T\to 0$ and $p/\sqrt{T^{\kappa-2}\log^{\kappa} p}\to 0$; (v) $\|\mathbf{m}-\mathbf{Z}\rho\|_{NT} = o_P(1/\sqrt{NT})$; (vi) for every $j,l\in[p]$ and $k\geq 0$, the $\tau$-mixing coefficients of $(u_{i,t}u_{i,t+k}x_{i,t,j}x_{i,t+k,l})_{t\in\Z}$ are $\check\tau_t\leq ct^{-d}$ for some universal constants $c>0$ and $d>1$; (vi) for each $i$, $\{(u_{i,t},z_{i,t}^\top)^\top:t\in\Z\}$ is a stationary process that is also i.i.d. over $i$, Assumption~\ref{as:data} holds with $a>(q-1)/(q-2)\vee (q\delta+1)/(q-2-\delta)$ with $q>2+\delta$ and $\delta>0$.
\end{assumption}
Assumption~\ref{as:clt} (i) requires that the conditional variance of the regression error is bounded. Condition (ii) requires that the rows of the precision matrix have bounded $\ell_1$ norm and is a plausible assumption in the high-dimensional setting, where the inverse covariance matrix is often sparse. Condition (iii) is a mild restriction needed for the consistency of the sample variance of regression errors. The rate conditions in (iv) are similar to the condition used in \cite{babii2020inference}. Lastly, condition (v) is trivially satisfied when the projection coefficients are sparse and, more generally, it requires that the misspecification error vanishes asymptotically sufficiently fast.

\smallskip

The following result describes a large-sample approximation to the distribution of the debiased sg-LASSO estimator with serially correlated heavy-tailed errors.
\begin{theorem}\label{thm:pooled_clt}
	Suppose that Assumptions~\ref{as:data}, \ref{as:covariance}, \ref{as:tuning}, \ref{as:rates}, and \ref{as:clt} are satisfied for the sg-LASSO regression and for each nodewise LASSO regression $j\in G$. Then 
	\begin{equation*}
	\sqrt{NT}(\hat\rho_G + B_G - \rho_G) \xrightarrow{d} N(0,\Xi_G)
	\end{equation*}
	with the long-run variance $\Xi_G = \lim_{T\to\infty}\Var\left(\frac{1}{\sqrt{T}}\sum_{t=1}^T u_{i,t}\Theta_Gz_{i,t}\right)$.
\end{theorem}

Theorem~\ref{thm:pooled_clt} applies to panel data consisting of non-Gaussian, heavy-tailed, and persistent time series under the large $N$ and $T$ large sample approximation. In contrast to the fixed $T$ approximations, Theorem~\ref{thm:pooled_clt} leads to more precise inference, e.g., the standard errors and the length of confidence intervals would scale at $O(1/\sqrt{NT})$ rate instead of $O(1/\sqrt{N})$ that we typically encounter for fixed $T$ approximations.

\smallskip

To estimate $\Xi_G$, we can use the following pooled HAC estimator
\begin{equation*}
	\hat\Xi_G = \frac{1}{N}\sum_{i=1}^N\sum_{|k|<T}K\left(\frac{k}{M_T}\right)\hat\Gamma_{k,i},
\end{equation*}
where $\hat{\Gamma}_{k,i} = \hat\Theta_G\left(\frac{1}{T}\sum_{t=1}^{T-k}\hat u_{i,t}\hat u_{i,t+k} x_{i,t}x_{i,t+k}^\top\right)\hat\Theta_G^\top$, $\hat u_{i,t}$ is the sg-LASSO residual, and $\hat\Gamma_{-k,i}=\hat{\Gamma}_{k,i}^\top$. The kernel function $K:\R\to[-1,1]$ with $K(0)=1$ is puts less weight on more distant noisy covariances, while $M_T\uparrow\infty$ is a bandwidth (or lag truncation) parameter; see \cite{babii2020inference} for more details as well as formal results on the validity of HAC-based inference using sg-LASSO residuals.

\section{Monte Carlo simulations \label{sec:mc}}

In this section, we assess the finite sample performance of the Granger causality tests for high-dimensional pooled panel data MIDAS regressions. A first subsection describes the design, followed by a second reporting the findings.

\subsection{Design}

We simulate the data from the following DGP: 
\begin{equation}
	\label{eq:DGP_MC}
	y_{i,t} = \alpha + \rho y_{i,t-1} + \sum_{k=1}^K\frac{1}{m}\sum_{j=1}^m \omega((j-1)/m;\beta_k)x_{i,t-(j-1)/m,k} + u_{i,t},
\end{equation} 
where $i\in[N]$, $t\in[T]$, $\alpha$ is the the common intercept, $\frac{1}{m}\sum_{j=1}^m \omega((j-1)/m;\beta_k)$ is the weight function for $k$-th high-frequency covariate and the error term is $u_{i,t} \sim_{i.i.d.}N(0,4)$. The DGP corresponds to the target variable of interest $y_{i,t}$ driven by one autoregressive lag augmented with high-frequency series. The DGP is therefore a pooled MIDAS panel data model.

\smallskip

We set $\rho=0.15$ and take the first high-frequency regressor, $k$ = 1, as relevant, i.e.\ the first regressor Granger causes the response variable. We are interested in quarterly/monthly data, and use four quarters of data for the high-frequency regressors so that $m$ = 12. The high-frequency regressors are generated as  $K$ i.i.d.\ realizations of univariate autoregressive (AR) processes $x_h = \rho x_{h-1}+ \varepsilon_h,$ where $\rho=0.7$ and  $\varepsilon_h\sim_{i.i.d.}N(0,1)$, where $h$ denotes the high-frequency sampling. For the DGP we rely on a commonly used weighting scheme in the MIDAS literature, namely the weights $\omega(s;\beta_k)$ for the only relevant high-frequency regressor $k=1$  determined by the beta density, $\mathrm{Beta}(3,3)$; see \cite{ghysels2007midas} or \cite{ghysels2019estimating}, for further details. The empirical estimation involves MIDAS regressions with Legendre polynomials of degree $L=3$. Lastly, we draw the intercepts $\alpha \sim \text{Uniform}(-4,4).$ Throughout the experiment, we fix the sample sizes to $T$ = 50 and $N$ = 30.

\smallskip

We compare the empirical size and power of the Granger causality test under different structures placed on the regression models. 

First, we compare sg-LASSO-MIDAS with LASSO-UMIDAS pooled panel data models. The former exploits the group structure of covariates by applying the sg-LASSO penalty function and a flexible way to model lags for each covariate using the MIDAS weight functions parametrized by low-dimensional coefficients. The latter pertains to the unstructured LASSO estimator together with the UMIDAS scheme. Introduced by \cite{foroni2015unrestricted}, UMIDAS consists of estimating a regression coefficient for each high-frequency lag separately, and therefore the weight function for each covariate is
\begin{equation}
	\label{eq:UMIDAS}
	\sum_{j=1}^m \omega((j-1)/m;\beta_k)x_{i,t-(j-1)/m,k} = \sum_{j=1}^m b_{j,k} x_{i,t-(j-1)/m,k}
\end{equation}
where $b_{j,k}$ is a regression coefficient associated with each high-frequency lag. We estimate regression coefficients by applying the standard unstructured LASSO estimator; hence we call the model LASSO-UMIDAS.

\smallskip

Second, we compare the pooled panel with individual time series regressions, for sg-LASSO-MIDAS and LASSO-UMIDAS, where the former exploits the benefits of the panel structure and the latter does not. In this case, we take the first sample $i=1$ to compute empirical size and power of the Granger test for the individual regression models. \cite{babii2020inference} propose tests of Granger causality in univariate regularized regressions and high-dimensional data.
\begin{table}[htbp]
	\centering
	{\scriptsize
		\begin{tabular}{r cc cc c cc cc}
			&\multicolumn{9}{c}{Pooled Panel}\\
			&\multicolumn{4}{c}{\textit{Parzen kernel}}&&\multicolumn{4}{c}{\textit{Quadratic spectral kernel}}\\\hline
			$M_T\backslash a$& 0 & 1/5 & 1/4 & 1/3 &  & 0 & 1/5 & 1/4 & 1/3 \\ 
			&\multicolumn{9}{c}{\underline{sg-LASSO-MIDAS}}\\
			10 & 0.051 & 0.835 & 0.959 & 0.999 &  & 0.056  & 0.841 & 0.963 & 0.998 \\ 
			20 & 0.049 & 0.822 & 0.954 & 0.999 &  & 0.047 & 0.828 & 0.957 & 0.998 \\ 
			30 & 0.046 & 0.803 & 0.953 & 0.999 &  & 0.047 & 0.823 & 0.956 & 0.998 \\ 
			&\multicolumn{9}{c}{\underline{LASSO-UMIDAS}}\\
			10 & 0.039 & 0.551 & 0.788 & 0.978  &  & 0.042 & 0.549 & 0.797 & 0.979 \\ 
			20 & 0.030 & 0.514 & 0.762 & 0.970 &  & 0.033 & 0.535 & 0.780 & 0.977 \\ 
			30 & 0.021 & 0.494 & 0.735 & 0.964 &  & 0.025 & 0.514 & 0.758 & 0.972 \\ 
			\hline
			&&&&&&&&&\\
			&\multicolumn{9}{c}{Individual Regressions}\\
			&\multicolumn{4}{c}{\textit{Parzen kernel}}&&\multicolumn{4}{c}{\textit{Quadratic spectral kernel}}\\\hline
			$M_T\backslash a$& 0 & 1/5 & 1/4 & 1/3 &  & 0 & 1/5 & 1/4 & 1/3 \\ 
			&\multicolumn{9}{c}{\underline{sg-LASSO-MIDAS}}\\
			10 & 0.090 & 0.356 & 0.406 & 0.548 &  & 0.094 & 0.349 & 0.356 & 0.486 \\ 
			20 & 0.097 & 0.345 & 0.406 & 0.548 &  & 0.094 & 0.350 & 0.360 & 0.492 \\ 
			30 & 0.092 & 0.345 & 0.403 & 0.547 &  & 0.093 & 0.356 & 0.379 & 0.524 \\ 
			&\multicolumn{9}{c}{\underline{LASSO UMIDAS}}\\
			10 & 0.110 & 0.201 & 0.228 & 0.362 &  & 0.107 & 0.210 & 0.236 & 0.378 \\ 
			20 & 0.111 & 0.240 & 0.272 & 0.406 &  & 0.108 & 0.212 & 0.206 & 0.388 \\
			30 & 0.107 & 0.245 & 0.370 & 0.494 &  & 0.105 & 0.204 & 0.206 & 0.386 \\ 
			\hline
		\end{tabular}
		\caption{\small HAC-based inference simulation results --- We report results for a set of bandwidth parameters, denoted $M_T,$ and two kernel functions. \label{tab:infer_simul_gaussian}} 
	} 
\end{table}
\subsection{Simulation results}

In Table \ref{tab:infer_simul_gaussian}, we report the empirical rejection frequency (ERF) for the Granger causality test based on the HAC estimator with two different kernel functions, Parzen and Quadratic spectral, and two different estimation strategies, sg-LASSO-MIDAS and LASSO-UMIDAS. We test whether the first high-frequency covariate Granger causes the low-frequency series, which corresponds to the DGP potential causal pattern.  We report results for a set of bandwidth parameters, denoted $M_T$ = 10, 20 and 30. The reported results are based on 2000 Monte Carlo replications.

\smallskip

To assess the performance we scale the Beta density function by multiplying it with a constant $a\in\{0,1/5,1/4,1/3\}$, i.e.\ the weight function for the relevant covariate is:
\begin{equation*}
	a\frac{1}{m}\sum_{j=1}^m \omega((j-1)/m;\beta_k)
\end{equation*}
For $a=0$, the ERF shows the empirical size of the test for the nominal level of 5\%, while $a\in\{1/5,1/4,1/3\}$ the ERF shows the empirical power of the Granger causality test. For the larger scaling constant $a,$ the alternatives are separated further away from the null hypothesis and the Granger causality test is expected to perform better. 

\smallskip

The results reported in Table~\ref{tab:infer_simul_gaussian} show that the Granger causality test based on the sg-LASSO-MIDAS has empirical size close to the nominal level of $5\%$. In contrast, the LASSO-UMIDAS leads to undersized Granger causality tests with size distortions around 0.01. The Granger causality test based on the sg-LASSO-MIDAS has also better empirical power against each of the alternative hypotheses $a\in\{1/5,1/4,1/3\}$. Additionally, it approaches 1 much faster as opposed to the LASSO-UMIDAS.

\smallskip

The results for individual regressions reveal worse performance compared to pooled panel data regressions, hence showing the usefulness of pooling the data. The empirical size shows considerable size distortions of around 0.05. Tests for individual regressions have worse power compared to the pooled panel data cases. Nonetheless, similar to the pooled panel data cases, the sg-LASSO-MIDAS estimation method seems to have better empirical power when comparing to LASSO-UMIDAS.

\smallskip

Overall, the results of the Monte Carlo experiments indicate that the structured regularization leads to better Granger causality tests in small samples and that pooling individual series improves the results even further. 

\section{Do analysts leave money on the table? \label{sec:emp}}

In this section we revisit a topic raised by \cite{ball2018automated} and \cite{carabias2018real}. Their empirical findings suggest that analysts tend to focus on their firm/industry when making earnings predictions while not fully taking into account the impact of macroeconomic events. While their findings were suggestive, there was no formal testing in a data-rich environment. The theory established in the previous sections allows us to do so.

\smallskip

More specifically, we consider the earnings of 210 US firms using a set of predictors sampled at mixed frequencies --- quarterly, monthly and daily series. We use 26 predictors (and their lags), including traditional macro and financial series as well as non-standard series generated by textual analysis of financial news. 

\subsection{Data description}

The full sample consists of observations between the \(1^{st}\) of January, 2000 and the \(30^{th}\) of June, 2017. Due to the lagged dependent variables in the models, our effective sample starts at the third fiscal quarter of 2000. We collected data from CRSP and I/B/E/S  to compute quarterly earnings and firm-specific financial covariates; RavenPack was used to compute daily firm-level textual-analysis-based data; real-time monthly macroeconomic series are from the ALFRED; FRED is used to compute daily financial markets data and, lastly, monthly news attention series extracted from the {\it Wall Street Journal} articles were retrieved from \cite{bybee2019structure}.\footnote{The dataset is publicly available at \href{http://www.structureofnews.com/}{http://www.structureofnews.com/}.} Table \ref{tab:data} provides a list of the variables used in our analysis, whereas Online Appendix Section \ref{appsec:data-description} covers a detailed description of the RavenPack data. Finally, the list of all firms we consider in our analysis appears in Online Appendix Table \ref{tab:list_firms}. Table \ref{tab:data} has six panels, namely three panels of firm-level series: A1 -- describes earnings data, B1 -- describes daily firm-level stock market data, and C1 --  describes daily firm-level sentiment data series. The remaining three panels are: A2 -- describes real-time monthly macro series, B2 -- describes daily financial markets data, and C2 -- describes monthly news attention series. In the models we include 365 daily lags, 12 monthly lags and 4 quarterly lags respectively.

\subsection{Granger causality tests}

Whether analysts leave money on the table amounts to testing whether forecast errors in earnings can be predicted by current information variables. Hence, this amounts to performing something akin to the Granger causality test. In our empirical application we are dealing with a panel, and it is important to exploit the multivariate data structure to perform such tests.

\smallskip

We analyze the difference between realized earnings and analysts' predictions, i.e., the response variable $y_{i,t+1}$ is computed by taking the difference between realized earnings, denoted $e_{i,t+1}$, and the median of analysts' predictions for the quarter $t+1$, denoted $f_{i,t+1|t}$,
\begin{equation*}
	y_{i,t+1} = e_{i,t+1} - f_{i,t+1|t}. 
\end{equation*}
We then fit the following pooled panel data MIDAS model using sg-LASSO estimator:
\begin{equation*}\label{eq:ppgc}
	y_{i,t+1} = \alpha + \rho y_{i,t} + \sum_{k=1}^{K}\psi(L^{1/m};\beta_k)x_{i,t,k} + u_{i,t+1}.
\end{equation*}
We test which factors Granger cause future errors of earnings forecasts made by the analysts. In the sg-LASSO, groups are defined as all lags of a single  covariate $k;$ Legendre polynomials up to degree three are applied to all weight functions $\psi(L^{1/m};\beta_k).$ We use 10-fold cross-validation to tune both $\lambda$ and $\gamma$, where we define folds as adjacent blocks over the time series dimension to take into account the time series dependence. Similarly, we estimate the precision matrix using nodewise LASSO regressions selecting the tuning parameter in a similar vein. The results are reported  in Table \ref{tab:gcpanel}.

\smallskip

In Panel (A) of Table \ref{tab:gcpanel} we find that the AR(1) lag is significant, leading us to conclude that the prediction errors made by the analysts are persistent. The autoregressive coefficient is significant throughout all specifications of the models, including in a simple pooled AR(1) model. In the latter case, the AR(1) coefficient is estimated to be 0.147. 

\smallskip

Panel (B) of Table \ref{tab:gcpanel} reports that beyond the AR(1) we find that the highly significant covariates are TED rate, CPI inflation and real GDP growth. These results support previous findings that analysts tend to miss information associated with macroeconomic conditions --- including real GDP growth and the TED spread, which is an indicator of measure credit risk. The latter is rather surprising, as it indicates that analysts tend to miss out on credit risk information at the macro level in their earnings forecasts. Lastly, the term spread (10-year less 3-month treasury yield), often viewed as a business cycle indicator, is also significant at the 10\% level.

\smallskip

Finally, in Panel (C) of Table \ref{tab:gcpanel} we report results based on the unstructured LASSO applying UMIDAS for the lag polynomials of each covariate. The findings reveal similar results for the TED rate, but notably miss real GDP and CPI inflation as significant covariates.

\smallskip

In Table \ref{tab:gcpanel_largesmalldis} we show results based on a different way of pooling analysts' prediction errors $y_{i,t+1}$. We split the data into two parts based on how large the average disagreement among analysts is. For each firm, we compute the forecast disagreement as the difference between 95\% and 5\% percentile of the empirical forecast distribution and take the average over the sample. We sort from high to low disagreement and split the sample of firms into two subsamples of equal size. The results show that macro variables which are significant for the full sample are also significant for the large disagreement subsample. On the other hand, little significance is reported for the low disagreement subsample. In this case, only the AR(1) lag and stock returns are significant at the 5\% significance level.  

\smallskip

Lastly, in Figure \ref{fig:granger_causality_ind} we plot the ratio of firms for which we find Granger causality based on individual regressions versus panel models. In Panel (a) we plot the ratios for sg-LASSO estimator using MIDAS weighting scheme while in Panel  (b) we plot the ratios for the LASSO estimator with UMIDAS scheme. The plot shows ratios for each covariate representing the fraction with respect to sg-LASSO (Panel (a)) or LASSO with UMIDAS (Panel (b)) each covariate is significant by running individual regressions. For example, the AR(1) lag is significant for around 30\% (0.3) of firms when running individual sg-LASSO-MIDAS regressions. Some covariates that are not significant in pooled panels are significant for some firms; therefore, we show results for all covariates, including those that are not significant in pooled panel cases. We also show how the ratios differ for low (dark-gray color) versus  high disagreement (light-gray color) firms. They represent whether a specific firm we run an individual regression for is in the high-disagreement versus low-disagreement subsample. 
Interestingly, the largest ratios are for AR(1), TED rate, Real GDP, CPI inflation and term spread in the case of sg-LASSO-MIDAS. Moreover, the portion of firms in the high disagreement subsample seem to have the largest ratios. In the case of LASSO-UMIDAS, the ratios show a less clear pattern, with only the AR(1) and TED rate covariates significant for a larger number of firms.

\section{Conclusions \label{sec:conclusion}}
This paper introduced a new class of high-dimensional panel data regression models with dictionaries and sg-LASSO regularization. This type of regularization is an especially attractive choice for predictive panel data regressions, where the low- and/or the high-frequency lags define a clear group structure. The estimator nests the LASSO and the group LASSO estimators as special cases. Our theoretical treatment allows for heavy-tailed data frequently encountered in financial time series. To that end, we obtain a new panel data concentration inequality of the Fuk-Nagaev type for $\tau$-mixing processes, which allows us to establish oracle inequalities that are used subsequently to develop the debiased HAC inference for the panel data sg-LASSO estimator.

\smallskip

Using the theory of HAC-based inference for pooled panel data regressions developed in our paper, our empirical analysis revisits a topic raised by earlier literature that analysts tend to focus on firm and/or industry information when forming earnings forecasts, while not fully taking into account the macroeconomic data.  Our results suggest that indeed analysts tend to miss on macro information, i.e., macro variables turn out to be significant in pooled panel regression models.

\newpage
\begin{landscape}
	\begin{table}
		\centering
		{\scriptsize
			\begin{tabular}{r |l ccc}
				& id & Frequency & Source  & T-code \\ 
				\hline
				\hline
				\multicolumn{5}{c}{\underline{Firm-level series}} \\
				\multicolumn{5}{c}{Panel A1.} \\
				- & Earnings & quarterly & CRSP \& I/B/E/S & 1 \\
				- & Earnings consensus forecasts & quarterly & CRSP \& I/B/E/S & 1 \\
				- & Other earnings/earnings forecast implied series & quarterly & CRSP \& I/B/E/S & 1 \\
				\multicolumn{5}{c}{Panel B1. } \\ 
				1 & Stock returns & daily & CRSP & 1 \\ 
				2 & Realized variance measure & daily & CRSP/computations & 1 \\ 
				\multicolumn{5}{c}{Panel C1. } \\ 
				3 & Event Sentiment Score (ESS) & daily & RavenPack & 1 \\
				4 & Aggregate Event Sentiment (AES) & daily & RavenPack & 1 \\
				5 & Aggregate Event Volume (AEV) & daily & RavenPack & 1 \\
				6 & Composite Sentiment Score (CSS) & daily & RavenPack & 1 \\
				7 & News Impact Projections (NIP) & daily & RavenPack & 1 \\
				\multicolumn{5}{c}{\underline{Other series}} \\
				\multicolumn{5}{c}{Panel A2. } \\
				8 & Industrial Production Index & monthly & ALFRED & 3 \\
				9 & CPI inflation & monthly & ALFRED & 4 \\
				10 & Unemployment rate & monthly & ALFRED & 1 \\
				11 & Real GDP & quarterly & ALFRED & 2 \\
				\multicolumn{5}{c}{Panel B2. } \\
				12 & Crude Oil Prices & daily & FRED & 4 \\ 
				13 & S\&P 500 & daily & CRSP & 3 \\ 
				14 & VIX Volatility Index & daily & FRED & 1 \\ 
				15 & Moodys Aaa less 10-Year Treasury & daily & FRED & 1 \\ 
				16 & Moodys Baa less 10-Year Treasury & daily & FRED & 1 \\ 
				17 & Moodys Baa less Aaa (corporate yield spread) & daily & FRED & 1 \\ 
				18 & 10-Year Treasury minus 3-Month Treasury (term spread) & daily & FRED & 1 \\ 
				19 & 3-Month Treasury minus Effective Federal funds rate (short-term spread) & daily & FRED & 1 \\ 
				20 & TED rate & daily & FRED & 1 \\ 
				\multicolumn{5}{c}{Panel C2.} \\ 
				21 & Earnings & monthly & \cite{bybee2019structure} & 1 \\
				22 & Earnings forecasts & monthly & \cite{bybee2019structure} & 1 \\
				23 & Earnings losses & monthly & \cite{bybee2019structure} & 1 \\
				24 & Recession & monthly & \cite{bybee2019structure} & 1 \\
				25 & Revenue growth & monthly & \cite{bybee2019structure}& 1 \\
				26 & Revised estimate & monthly & \cite{bybee2019structure} & 1 \\
				\hline 
			\end{tabular}
		}
		\caption{\footnotesize Firm-level data description table --- The \textit{id} column gives mnemonics according to data source, which is given in the second column \textit{Source}. The column \textit{frequency} states the sampling frequency of the variable. The column \textit{T-code} denotes the data transformations applied to a time series, which are (1) not transformed, (2) $100[(x_t/x_{t-1})^4-1]$, (3) $\Delta$ log ($x_t$), (4) $\Delta^2$ log ($x_t$). The block of firm-level series contains three panels: A1 -- describes earnings data, B1 -- describes daily firm-level stock market data, and C1 -- describes daily firm-level sentiment data series. The block labeled "other series" also has three panels: A2 -- describes real-time monthly macro series, B2 -- describes daily financial markets data, and C2 -- describes monthly news attention series. In the models we include 365 daily lags, 12 monthly lags and 4 quarterly lags respectively. \label{tab:data}}
	\end{table}
\end{landscape}
\begin{table}
	\centering
	\begin{tabular}{r  rrrrrrr}
		Variable $\backslash M_T$ & 10 & 20 & 30 &  & 10 & 20 & 30  \\ \hline
		&\multicolumn{3}{c}{\underline{Quadratic Spectral}}&&\multicolumn{3}{c}{\underline{Parzen}} \\
		&\multicolumn{7}{c}{\underline{Panel (A) -- AR(1)}} \\
		AR(1) & 0.001 & 0.000 & 0.000 & & 0.002 & 0.001 & 0.001 \\
		&&&&&&& \\	
		&\multicolumn{7}{c}{\underline{Panel (B) -- sg-LASSO}} \\
		&\multicolumn{7}{c}{Significant variables at 5\% or less} \\
		AR(1) & 0.001 & 0.000 & 0.000 &  & 0.002 & 0.001 & 0.000 \\ 
		TED rate & 0.001 & 0.001 & 0.000 &  & 0.003 & 0.001 & 0.001 \\ 
		CPI inflation & 0.003 & 0.001 & 0.001 &  & 0.013 & 0.003 & 0.001 \\ 
		Real GDP & 0.028 & 0.003 & 0.001 &  & 0.035 & 0.021 & 0.006 \\ 
		&\multicolumn{7}{c}{Significant variables at 10\% level} \\
		Term spread & 0.012 & 0.014 & 0.023 &  & 0.053 & 0.016 & 0.015 \\ 
		&&&&&&& \\	
		&&&&&&& \\	
		&\multicolumn{7}{c}{\underline{Panel (C) -- LASSO (significant for sg-LASSO)}} \\
		&\multicolumn{7}{c}{Significant variables at 5\% or less} \\
		AR(1) & 0.001 & 0.000 & 0.000 &  & 0.002 & 0.001 & 0.000 \\ 
		TED rate & 0.000 & 0.000 & 0.000 &  & 0.000 & 0.000 & 0.000 \\ 
		CPI inflation & 0.677 & 0.390 & 0.461 &  & 0.651 & 0.724 & 0.576 \\ 
		Real GDP & 0.341 & 0.247 & 0.094 &  & 0.339 & 0.328 & 0.270 \\ 
		&\multicolumn{7}{c}{Significant variables at 10\% level} \\
		Term spread & 0.273 & 0.060 & 0.022 &  & 0.235 & 0.387 & 0.365 \\ 
		&&&&&&& \\	 
		&\multicolumn{7}{c}{\underline{LASSO (significant only for LASSO)}} \\
		&\multicolumn{7}{c}{Significant variables at 5\% or less} \\
		AAA less 10 year & 0.009 & 0.001 & 0.001 &  & 0.015 & 0.014 & 0.007 \\ 
		BAA less 10 year & 0.000 & 0.000 & 0.000 &  & 0.000 & 0.000 & 0.000 \\ 
		\hline
	\end{tabular}
	\caption{Significance testing results --- We report p-values for the AR(1) in Panel (A) and for the sg-LASSO using the MIDAS scheme with Legendre polynomials in Panel (B) displaying series significant at the 5\% or 10\% significance level. We also report results for the standard LASSO estimator together with the UMIDAS scheme in Panel (C). The results are reported for a range of bandwidth parameters ($M_T$ = 10, 20 and 30) and two kernel functions (Quadratic Spectral and Parzen).}
	\label{tab:gcpanel}
\end{table}

\begin{table}
	\centering
	\begin{tabular}{r  rrrrrrr}
		Variable $\backslash M_T$ & 10 & 20 & 30 &  & 10 & 20 & 30  \\ \hline
		&\multicolumn{3}{c}{\underline{Quadratic Spectral}}&&\multicolumn{3}{c}{\underline{Parzen}} \\
		&\multicolumn{7}{c}{\underline{Large disagreement}} \\
		&\multicolumn{7}{c}{Significant variables at 5\% or less} \\
		AR(1) & 0.002 & 0.001 & 0.000 &  & 0.004 & 0.001 & 0.001 \\ 
		Term spread & 0.029 & 0.023 & 0.016 &  & 0.085 & 0.036 & 0.026 \\ 
		TED rate & 0.002 & 0.001 & 0.001 &  & 0.016 & 0.002 & 0.001 \\ 
		CPI inflation & 0.016 & 0.009 & 0.007 &  & 0.040 & 0.018 & 0.011 \\ 
		&\multicolumn{7}{c}{Significant variables at 10\% level} \\
		Real GDP & 0.098 & 0.005 & 0.000 &  & 0.098 & 0.082 & 0.021 \\ 
		&&&&&&& \\	
		&\multicolumn{7}{c}{\underline{Small disagreement}} \\
		&\multicolumn{7}{c}{Significant variables at 5\% or less} \\
		AR(1) & 0.000 & 0.000 & 0.000 &  & 0.000 & 0.000 & 0.000 \\ 
		Stock returns & 0.008 & 0.004 & 0.003 &  & 0.015 & 0.008 & 0.006 \\ 
		&\multicolumn{7}{c}{Significant variables at 10\% level} \\
		Unemployment rate & 0.060 & 0.043 & 0.045 &  & 0.060 & 0.056 & 0.048 \\ 
		\hline
	\end{tabular}
	\caption{Significance testing results --- We report p-values for the AR(1) and for the sg-LASSO-MIDAS models, displaying series significant at the 5\% or 10\% significance level. The results are reported for a range of bandwidth parameters and two kernel functions. We pool the response based on large versus small disagreement, which we measure as the average (over time series) of the difference between 95\% and 5\% percentile of the empirical forecast distribution of the analysts.}
	\label{tab:gcpanel_largesmalldis}
\end{table}

\begin{figure}[htp!]
	\centering
	\begin{subfigure}{0.49\textwidth} 
		\includegraphics[width=\textwidth]{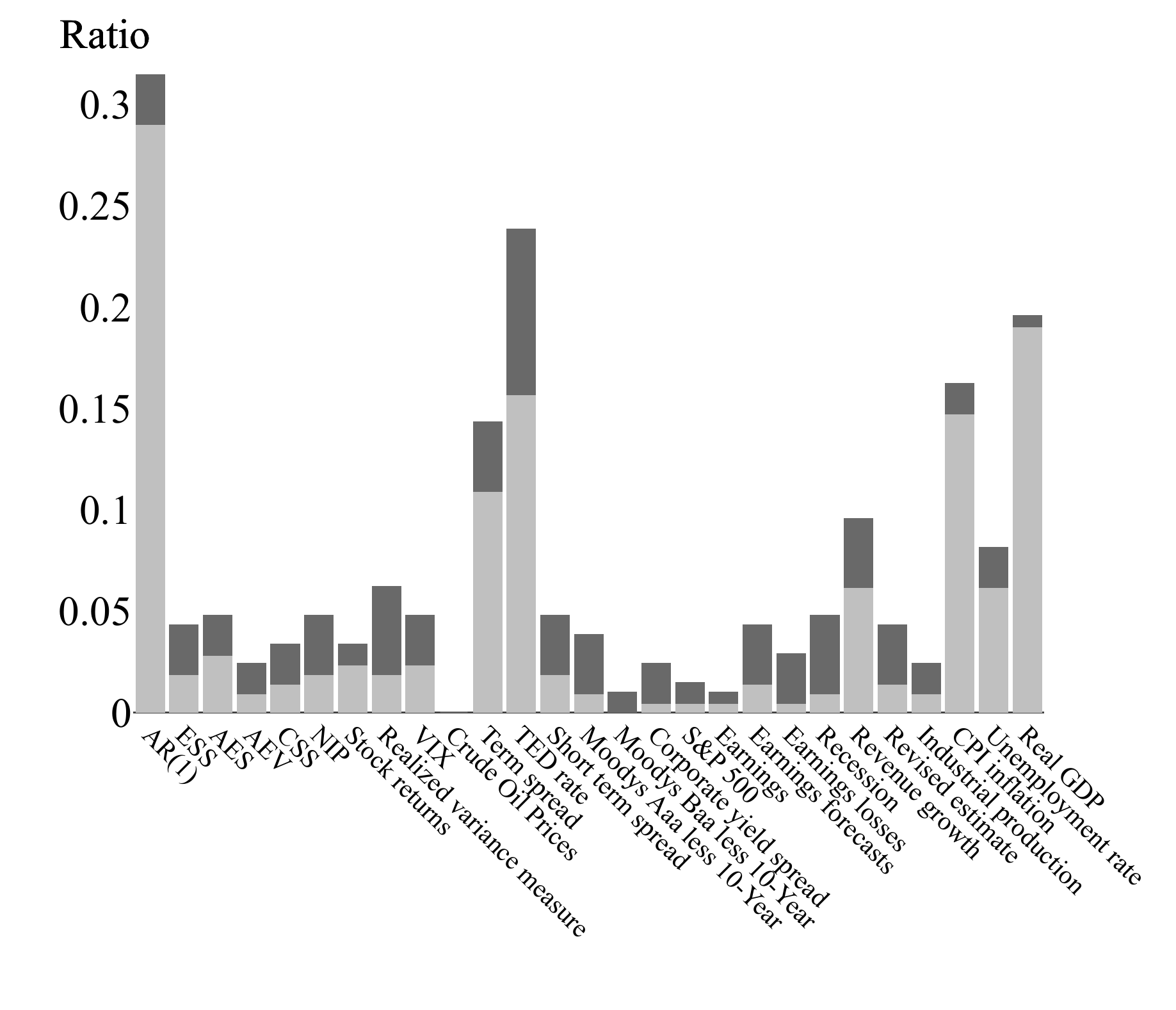}
		\caption{sg-LASSO-MIDAS} 
	\end{subfigure}
	\begin{subfigure}{0.49\textwidth} 
		\includegraphics[width=\textwidth]{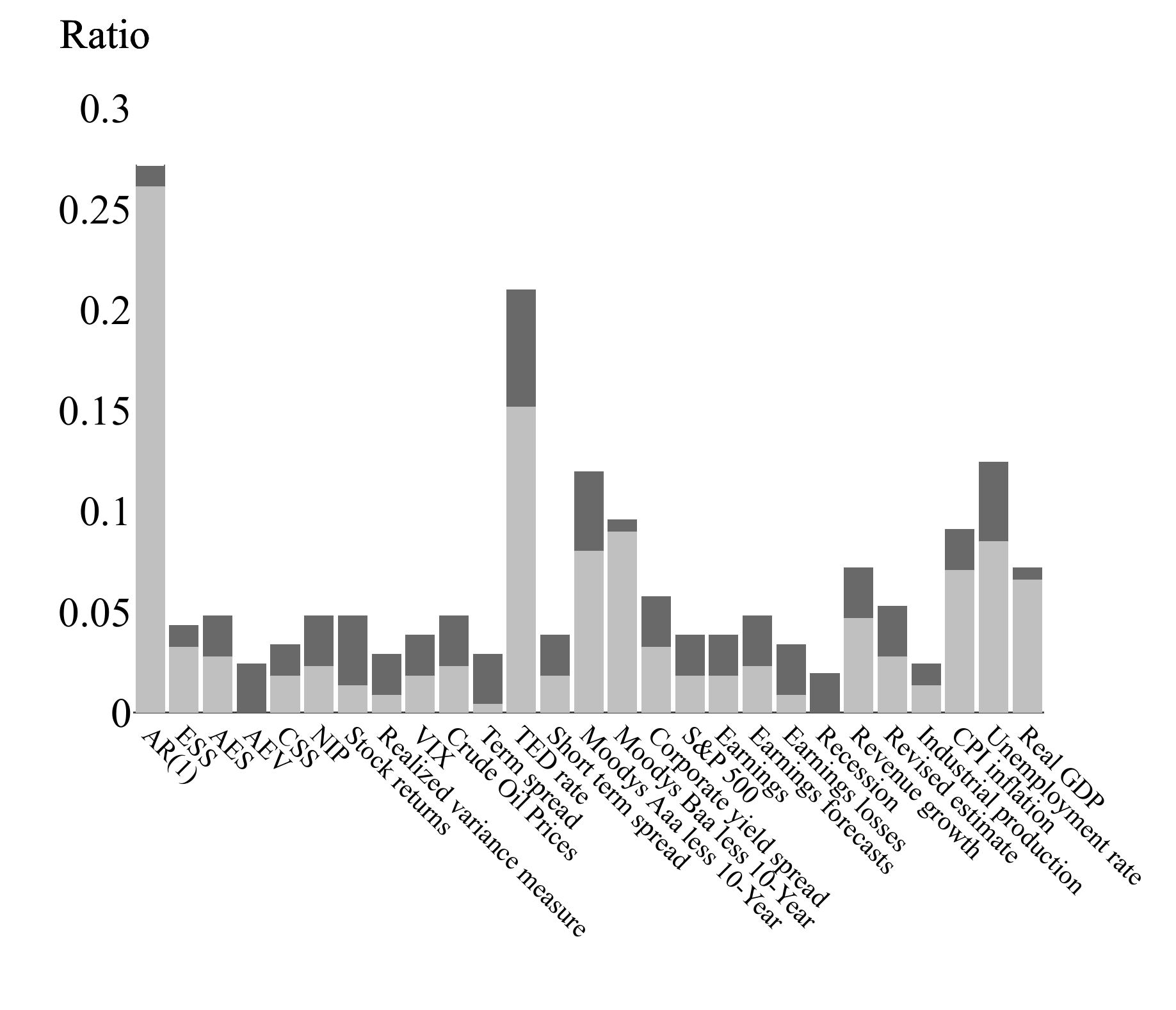}
		\caption{LASSO-UMIDAS} 
	\end{subfigure}
	\caption{Individual regression-based Granger causality tests. In Panel (a) we plot the ratios based on sg-LASSO estimator and MIDAS weighting scheme with Legendre polynomials, while in Panel (b) we plot for the ratios for the standard LASSO estimator with UMIDAS weighting scheme. The lighter-gray color shows the ratio for firms with high disagreement, while the dark-gray color shows the ratio for firms with low disagreement; see Table \ref{tab:gcpanel_largesmalldis}. All results are based on the 5\% significance level. \label{fig:granger_causality_ind}}
\end{figure}

\clearpage

\bibliographystyle{econometrica}
\bibliography{midas_ml}

\clearpage
\appendix
\setcounter{section}{0}
\setcounter{equation}{0}
\setcounter{table}{0}
\setcounter{figure}{0}
\renewcommand{\theequation}{A.\arabic{equation}}
\renewcommand\thetable{A.\arabic{table}}
\renewcommand\thefigure{A.\arabic{figure}}
\renewcommand\thesection{A.\arabic{section}}
\renewcommand\thetheorem{A.\arabic{theorem}}

\begin{center}
	{\LARGE\textbf{APPENDIX}}	
\end{center}
\appendix
{\small
\section{Proofs\label{appsec:proofs}}
\begin{proof}[Proof of Theorem~\ref{thm:pooled_panel}]
	By Fermat's rule, the pooled sg-LASSO satisfies
	\begin{equation*}
		\mathbf{Z}^\top(\mathbf{Z}\hat\rho - \mathbf{y})/NT + \lambda z^* = 0_{p+1}
	\end{equation*}
	for some $z^*\in\partial\Omega(\hat\rho)$, where $\partial\Omega(\hat\rho)$ is the subdifferential of $b\mapsto\Omega(b)$ at $\hat\rho$. Taking the inner product with $\rho - \hat\rho$
	\begin{equation*}
	\begin{aligned}
		\langle \mathbf{Z}^\top(\mathbf{y} - \mathbf{Z}\hat\rho),\rho-\hat\rho\rangle_{NT} & = \lambda \langle z^*,\rho-\hat\rho\rangle\\
		& \leq \lambda\left\{\Omega(\rho) - \Omega(\hat\rho)\right\},
	\end{aligned}
	\end{equation*}
	where the last line follows from the definition of the subdifferential. Since $\mathbf{y}=\mathbf{m}+\mathbf{u}$, the inequality can be rewritten as
	\begin{equation*}
	\begin{aligned}
		\|\mathbf{Z}(\hat\rho - \rho)\|^2_{NT}-\lambda\left\{\Omega(\rho) - \Omega(\hat\rho)\right\} & \leq \langle\mathbf{Z}^\top(\mathbf{Z}\rho - \mathbf{y}),\rho-\hat\rho\rangle_{NT} \\
		& =  \langle\mathbf{Z}^\top\mathbf{u},\hat\rho-\rho\rangle_{NT} + \langle\mathbf{m} - \mathbf{Z}\rho,\mathbf{Z}(\hat\rho-\rho)\rangle_{NT}.
	\end{aligned}
	\end{equation*}
	By the dual norm inequality $\langle\mathbf{Z}^\top\mathbf{u},\hat\rho-\rho\rangle_{NT} \leq \Omega^*(\mathbf{Z}^\top\mathbf{u}/NT) \Omega(\hat\rho-\rho)$, where $\Omega^*$ is the dual norm of $\Omega$. Then by \cite{babii2020machine}, Lemma A.2.1
	\begin{equation*}
	\begin{aligned}
		\Omega^*(\mathbf{Z}^\top\mathbf{u}/NT) & \leq \gamma|\mathbf{Z}^\top\mathbf{u}/NT|_\infty + (1-\gamma)\max_{G\in\mathcal{G}}|\mathbf{Z}^\top_G\mathbf{u}/NT|_{2} \\
		& \leq \max_{G\in\mathcal{G}}\sqrt{|G|}|\mathbf{Z}^\top\mathbf{u}/NT|_\infty \\
		& \leq \lambda/c, 
	\end{aligned}
	\end{equation*}
	where the last line follows from Theorem~\ref{thm:fn_long} with probability at least $1-\delta$ and Assumption~\ref{as:tuning} for some $c>1$. Therefore,
	\begin{equation}\label{eq:twopoint2}
		\|\mathbf{Z}\Delta\|^2_{NT} - \lambda\left\{\Omega(\rho) - \Omega(\hat\rho)\right\} \leq \frac{\lambda}{c}\Omega(\Delta) + \|\mathbf{m} - \mathbf{Z}\rho\|_{NT}\|\mathbf{Z}\Delta\|_{NT} \text{ with } \Delta = \hat\rho - \rho.
	\end{equation}
Note that the sg-LASSO penalty function can be decomposed as a sum of two semi-norms $\Omega(r) = \Omega_0(r) + \Omega_1(r),\forall r\in\R^{1+p}$ with
	\begin{equation*}
	\Omega_0(r) = \gamma|r_{S_0}|_1 + (1-\gamma)\sum_{G\in\mathcal{G}_0}|r_G|_2\quad\text{and}\quad \Omega_1(r) = \gamma|r_{S_0^c}|_1 + (1-\gamma)\sum_{G\in\mathcal{G}_0^c}|r_G|_2.
	\end{equation*}
	Note also that $\Omega_1(\rho) = 0$ and $\Omega_1(\hat\rho) = \Omega_1(\hat\rho - \rho)$. Then
	\begin{equation}\label{eq:cone0}
	\begin{aligned}
	\Omega(\rho) - \Omega(\hat\rho) & = \Omega_0(\rho) - \Omega_0(\hat\rho) - \Omega_1(\hat\rho) \\
	& \leq \Omega_0(\hat\rho - \rho) - \Omega_1(\hat\rho - \rho)  = \Omega_0(\Delta) - \Omega_1(\Delta) .
	\end{aligned}
	\end{equation}
Suppose that $\|\mathbf{m} - \mathbf{Z}\rho\|_{NT} \leq \frac{1}{2}\|\mathbf{Z}\Delta\|_{NT}$. Then it follows from equations (\ref{eq:twopoint2}) and (\ref{eq:cone0}) that
	\begin{equation*}
	\begin{aligned}
		\|\mathbf{Z}\Delta\|_{NT}^2 & \leq 2\frac{\lambda}{c}\Omega(\Delta) + 2\lambda\left\{\Omega_0(\Delta) - \Omega_1(\Delta)\right\} \\
		& = 2\frac{\lambda}{c}\left\{\Omega_1(\Delta) + \Omega_0(\Delta) \right\} + 2\lambda\left\{\Omega_0(\Delta) - \Omega_1(\Delta)\right\} \\
	\end{aligned}
	\end{equation*}
	Since the left side of this equation is greater or equal to zero, this shows that
	\begin{equation}\label{eq:cone_condition}
	\Omega_1(\Delta) \leq \frac{c+1}{c-1}\Omega_0(\Delta).
	\end{equation}
	Put $\Sigma_{N,T}=\frac{1}{NT}\sum_{i=1}^N\sum_{t=1}^T\E[z_{i,t}z_{i,t}^\top]$. Therefore,
	\begin{equation*}
	\begin{aligned}
	\Omega(\Delta) & \leq \frac{2c}{c-1}\Omega_0(\Delta) \leq \frac{2c}{c-1}\sqrt{s|\Delta|^2_2} \leq \frac{2c}{c-1}\sqrt{\frac{s}{\gamma_{\min}}|\Sigma_{N,T}^{1/2}\Delta|^2_2} \\
	& = \frac{2c}{c-1}\sqrt{\frac{s}{\gamma_{\min}}\left\{\|\mathbf{Z}\Delta\|_{NT}^2 + \Delta^\top(\hat\Sigma - \Sigma_{N,T})\Delta \right\} }\\
	& \leq \frac{2c}{c-1}\sqrt{\frac{s}{\gamma_{\min}}\left\{\|\mathbf{Z}\Delta\|_{NT}^2 + \Omega(\Delta)\Omega^*((\hat\Sigma - \Sigma_{N,T})\Delta) \right\}} \\
	& \leq  \frac{2c}{c-1}\sqrt{\frac{s}{\gamma_{\min}}\left\{2(1+c^{-1})\lambda\Omega(\Delta) + \Omega^2(\Delta)G^*|\mathrm{vech}(\hat\Sigma - \Sigma_{N,T})|_\infty \right\}},
	\end{aligned}
	\end{equation*}
	where we set $G^*=\max_{G\in\mathcal{G}}\sqrt{|G|}$ and use H\"{o}lder's inequality, inequalities in equations (\ref{eq:twopoint2}) and (\ref{eq:cone_condition}), Assumption~\ref{as:covariance}, $\hat \Sigma = \mathbf{Z}^\top\mathbf{Z}/NT$, and \cite{babii2020machine}, Lemma A.2.1. This shows that with probability at least $1-\delta$
	\begin{equation}\label{eq:Delta_inequality}
		\Omega(\Delta) \leq \frac{4c^2s}{(c-1)^2\gamma_{\min}}\left\{2(1+c^{-1})\lambda + \Omega(\Delta)G^*|\mathrm{vech}(\hat\Sigma - \Sigma_{N,T})|_\infty \right\}.
	\end{equation}
	Consider the following event $E = \{|\mathrm{vech}(\hat\Sigma - \Sigma_{N,T})|_\infty<(2c^*G^*s)^{-1}\}$ with $c^* = (3c+1)^2 / (\gamma_{\min}(c-1)^2)$, and note that under Assumption~\ref{as:data} by Theorem~\ref{thm:fn_long}
	\begin{equation*}
	\begin{aligned}
	\Pr(E^c) & = \Pr\left(\max_{1\leq j\leq k\leq p}\left|\frac{1}{NT}\sum_{i=1}^N\sum_{t=1}^Tz_{i,t,j}z_{i,t,k} - \E[z_{i,t,j}z_{i,t,k}]\right|\geq \frac{1}{2c^*G^*s}\right) \\
	& \lesssim p^2(NT)^{1-\tilde\kappa}s^{\tilde \kappa} + p^2e^{-cNT/s^2}
	\end{aligned}
	\end{equation*}	
	for some $c>0$. On the event $E$, the inequality in equation (\ref{eq:Delta_inequality}) implies $\Omega(\Delta) \lesssim s\lambda$, and whence from the equation (\ref{eq:twopoint2}) by the triangle inequality
	\begin{equation*}
	\begin{aligned}
	\|\mathbf{Z}\Delta\|_{NT}^2 \leq 2(1+c^{-1})\lambda\Omega(\Delta) \lesssim s\lambda^2.
	\end{aligned}
	\end{equation*}
	Therefore, we obtain the statement of the theorem as long as $\|\mathbf{m} - \mathbf{Z}\rho\|_{NT} \leq \frac{1}{2}\|\mathbf{Z}\Delta\|_{NT}$. Suppose now that $\|\mathbf{m} - \mathbf{Z}\rho\|_{NT} > \frac{1}{2}\|\mathbf{Z}\Delta\|_{NT}$. Then
	\begin{equation*}
	\|\mathbf{Z}\Delta\|^2_{NT}  \leq 4\|\mathbf{m} - \mathbf{Z}\rho\|^2_{NT}.
	\end{equation*}
	Therefore, the first statement of the theorem always holds with probability at least $1-\delta-O(r_{N,T}^{\rm pooled})$
	\begin{equation*}
	\|\mathbf{Z}\Delta\|^2_{NT}  \lesssim s\lambda^2 +  \|\mathbf{m} - \mathbf{Z}\rho\|^2_{NT}.
	\end{equation*}
For the second statement, suppose first that
	\begin{equation}\label{eq:cone2}
		\Omega_1(\Delta) \leq 2\frac{c+1}{c-1}\Omega_0(\Delta).
	\end{equation}
	Then by the same arguments as before, on the event $E$, we have
		\begin{equation*}
		\begin{aligned}
		\Omega(\Delta) & \leq \left(1 + 2\frac{c+1}{c-1}\right)\Omega_0(\Delta) \\
		& \leq \frac{3c+1}{c-1}\sqrt{\frac{s}{\gamma_{\min}}\left\{\|\mathbf{Z}\Delta\|_{NT}^2 + \frac{1}{2c^*s}\Omega^2(\Delta) \right\}} \\
		& = \sqrt{\frac{(3c+1)^2}{(c-1)^2\gamma_{\min}}s\|\mathbf{Z}\Delta\|^2_{NT} + \frac{1}{2}\Omega^2(\Delta)}
		\end{aligned}
		\end{equation*}
	or simply
	\begin{equation*}
	\begin{aligned}
		\Omega(\Delta)  \leq \sqrt{2}\frac{(3c+1)}{(c-1)}\sqrt{\frac{s}{\gamma_{\min}}}\|\mathbf{Z}\Delta\|_{NT} 
	 \lesssim s\lambda + \sqrt{s}\|\mathbf{m} - \mathbf{Z}\rho\|_{NT},
	\end{aligned}
	\end{equation*}	
	where we use the first statement of the theorem. On the other hand, if the inequality in equation (\ref{eq:cone2}) does not hold, then the inequality in equation (\ref{eq:cone_condition}) also does not hold, which implies that
	\begin{equation*}
		\|\mathbf{m}-\mathbf{Z}\rho\|_{NT} > \frac{1}{2}\|\mathbf{Z}\Delta\|_{NT}.
	\end{equation*}
	Then since $\|\mathbf{Z}\Delta\|_{NT}\geq 0$ from (\ref{eq:twopoint2}) we obtain
	\begin{equation*}
	\begin{aligned}
		0 & \leq \frac{1}{c}\Omega(\Delta) + \Omega(\rho) - \Omega(\hat\rho) + \frac{2}{\lambda}\|\mathbf{m} - \mathbf{Z}\rho\|^2_{NT} \\
		& \leq \frac{1}{c}\Omega(\Delta) + \Omega_0(\Delta) - \Omega_1(\Delta) + \frac{2}{\lambda}\|\mathbf{m} - \mathbf{Z}\rho\|^2_{NT},
	\end{aligned}
	\end{equation*}
	where we use equation (\ref{eq:cone0}). Since $\Omega(\Delta)=\Omega_1(\Delta) + \Omega_0(\Delta)$
	\begin{equation*}
	\begin{aligned}
		\Omega_1(\Delta) & \leq   \frac{c+1}{c-1}\Omega_0(\Delta) + \frac{2c}{\lambda(c-1)}\|\mathbf{m} - \mathbf{Z}\rho\|^2_{NT} \\
		& \leq  \frac{1}{2}\Omega_1(\Delta) + \frac{2c}{\lambda(c-1)}\|\mathbf{m} - \mathbf{Z}\rho\|^2_{NT},
	\end{aligned}
	\end{equation*}
	where we use the fact that the inequality in equation (\ref{eq:cone2}) does not hold. Therefore,
	\begin{equation*}
		\Omega_1(\Delta) \leq  \frac{4c}{\lambda(c-1)}\|\mathbf{m} - \mathbf{Z}\rho\|^2_{NT},
	\end{equation*}
	which shows that
	\begin{equation*}
		\Omega(\Delta) \lesssim \Omega_1(\Delta) \leq \frac{4c}{\lambda(c-1)}\|\mathbf{m} - \mathbf{Z}\rho\|^2_{NT}.
	\end{equation*}
	Therefore, with probability at least $1-\delta-O(r_{N,T}^{\rm pooled})$, we always have
	\begin{equation*}
		\Omega(\Delta) \lesssim s\lambda + \sqrt{s}\|\mathbf{m} - \mathbf{Z}\rho\|_{NT} + \frac{1}{\lambda}\|\mathbf{m} - \mathbf{Z}\rho\|^2_{NT}.
	\end{equation*}
	The result follows from the equivalence between $\Omega$ and $|.|_1$ norms provided that groups have fixed size.
\end{proof}

\begin{proof}[Proof of Theorem~\ref{thm:fixed_effects}]
	By Fermat's rule the solution to the fixed effects regression satisfies
	\begin{equation*}
	\mathbf{Z}^\top(\mathbf{Z}\hat\rho - \mathbf{y})/NT + \lambda z^* = 0_{N+p}, \text{ for some } z^* = \binom{0_N}{z^*_b}, 
	\end{equation*}
	where $0_N$ is $N$-dimensional vector of zeros, $z_b^*\in\partial\Omega(\hat\beta)$, $\hat\rho = (\hat\alpha^\top,\hat\beta^\top)^\top$, and $\partial\Omega(\hat\beta)$ is the sub-differential of $b\mapsto \Omega(b)$ at $\hat\beta$. Taking the inner product with $\rho - \hat\rho$
	\begin{equation*}
	\begin{aligned}
	\langle\mathbf{Z}^\top(\mathbf{y} - \mathbf{Z}\hat\rho),\rho - \hat\rho\rangle_{NT} & = \lambda\langle z^*,\rho-\hat\rho\rangle \\
	& = \lambda\langle z_b^*,\beta - \hat\beta\rangle  \leq \lambda\left\{\Omega(\beta) - \Omega(\hat\beta)\right\},
	\end{aligned}
	\end{equation*}
	where the last line follows from the definition of the sub-differential. Rearranging this inequality and using $\mathbf{y} = \mathbf{m} + \mathbf{u}$
	\begin{equation}\label{eq:two_point}
	\begin{aligned}
	\|\mathbf{Z}(\hat\rho - \rho)\|^2_{NT} - \lambda\left\{\Omega(\beta) - \Omega(\hat\beta)\right\} \leq &\langle\mathbf{Z}^\top\mathbf{u},\hat\rho-\rho\rangle_{NT} + \langle\mathbf{Z}^\top(\mathbf{m}-\mathbf{Z}\rho),\hat\rho - \rho\rangle_{NT} \\
	\leq &\langle B^\top\mathbf{u},\hat\alpha - \alpha\rangle_{NT} + \langle\mathbf{X}^\top\mathbf{u},\hat\beta - \beta\rangle_{NT} \\
	& \qquad + \|\mathbf{m} - \mathbf{Z}\rho\|_{NT}\|\mathbf{Z}(\hat\rho - \rho)\|_{NT} \\
	\leq &|B^\top\mathbf{u}/NT|_\infty|\hat\alpha - \alpha|_1 +  \Omega^*(\mathbf{X}^\top\mathbf{u}/NT)\Omega(\hat\beta - \beta) \\
	& \qquad + \|\mathbf{m} - \mathbf{Z}\rho\|_{NT}\|\mathbf{Z}(\hat\rho - \rho)\|_{NT} \\
	\leq &|B^\top\mathbf{u}/\sqrt{N}T|_\infty\vee\Omega^*(\mathbf{X}^\top\mathbf{u}/NT)\\
	&\times\left\{ |\hat\alpha - \alpha|_1/\sqrt{N} + \Omega(\hat\beta - \beta)\right\} \\
	& \qquad + \|\mathbf{m} - \mathbf{Z}\rho\|_{NT}\|\mathbf{Z}(\hat\rho - \rho)\|_{NT},
	\end{aligned}
	\end{equation}
	where the second line follows by the dual norm inequality and the Cauchy-Schwartz inequality, and $\Omega^*$ is the dual norm of $\Omega$. By \cite{babii2020machine}, Lemma A.2.1. and Theorem~\ref{thm:fn_long} under Assumption~\ref{as:data}, with probability at least $1-\delta/2$
	\begin{equation*}
		\Omega^*(\mathbf{X}^\top\mathbf{u}/NT)\leq \max_{G\in\mathcal{G}}\sqrt{|G|} |\mathbf{X}^\top\mathbf{u}/NT|_\infty \lesssim \left(\frac{p}{\delta(NT)^{\kappa-1}}\right)^{1/\kappa}\vee\sqrt{\frac{\log(16p/\delta)}{NT}}.
	\end{equation*}
	Similarly, under Assumption~\ref{as:data} by \cite{babii2020inference}, Theorem 3.1 with probability at least $1-\delta/2$
	\begin{equation*}
		|B^\top\mathbf{u}/\sqrt{N}T|_\infty = \max_{i\in[N]}\left|\frac{1}{\sqrt{N}T}\sum_{t=1}^Tu_{i,t}\right| \lesssim \left(\frac{N}{\delta N^{\kappa/2}T^{\kappa-1}}\right)^{1/\kappa}\vee\sqrt{\frac{\log(16N/\delta)}{NT}}.
	\end{equation*}
	Therefore, under Assumption~\ref{as:tuning_fe} with probability at least $1-\delta$
	\begin{equation*}
	|B^\top\mathbf{u}/NT|_\infty\vee\Omega^*(\mathbf{X}^\top\mathbf{u}/NT) \lesssim \left(\frac{(pN^{1-\kappa})\vee N^{1-\kappa/2}}{\delta T^{\kappa -1}}\right)^{1/\kappa}\vee \sqrt{\frac{\log(p\vee N/\delta)}{NT}} \lesssim \lambda.
	\end{equation*}
	In conjunction with the inequality in equation~(\ref{eq:two_point}), this gives
	\begin{equation}\label{eq:two_point2}
	\begin{aligned}
	\|\mathbf{Z}\Delta\|^2_{NT} & \leq c^{-1}\lambda\left\{|\hat\alpha - \alpha|_1/\sqrt{N} + \Omega(\hat\beta - \beta) \right\} \\
	& \quad  + \|\mathbf{m} - \mathbf{Z}\rho\|_{NT}\|\mathbf{Z}\Delta\|_{NT} +  \lambda\left\{\Omega(\beta) - \Omega(\hat\beta)\right\} \\
	& \leq (c^{-1} + 1)\lambda\left\{|\hat\alpha - \alpha|_1/\sqrt{N} + \Omega(\hat\beta - \beta) \right\} + \|\mathbf{m} - \mathbf{Z}\rho\|_{NT}\|\mathbf{Z}\Delta\|_{NT}
	\end{aligned}
	\end{equation}
	for some $c>1$ and $\Delta = \hat\rho - \rho$, where the second line follows by the triangle inequality. Note that the sg-LASSO penalty function can be decomposed as a sum of two semi-norms $\Omega(b) = \Omega_0(b) + \Omega_1(b),\forall b\in\R^p$ with
	\begin{equation*}
		\Omega_0(b) = \gamma|b_{S_0}|_1 + (1-\gamma)\sum_{G\in\mathcal{G}_0}|b_G|_2\quad\text{and}\quad \Omega_1(b) = \gamma|b_{S_0^c}|_1 + (1-\gamma)\sum_{G\in\mathcal{G}_0^c}|b_G|_2.
	\end{equation*}
	Note also that $\Omega_1(\beta) = 0$ and $\Omega_1(\hat\beta) = \Omega_1(\hat\beta - \beta)$. Then
	\begin{equation}\label{eq:cone_condition_0}
	\begin{aligned}
	\Omega(\beta) - \Omega(\hat\beta) & = \Omega_0(\beta) - \Omega_0(\hat\beta) - \Omega_1(\hat\beta) \\
	& \leq \Omega_0(\hat\beta - \beta) - \Omega_1(\hat\beta - \beta).
	\end{aligned}
	\end{equation}
Suppose that $\|\mathbf{m} - \mathbf{Z}\rho\|_{NT} \leq \frac{1}{2}\|\mathbf{Z}\Delta\|_{NT}$. Then from the first inequality in equation (\ref{eq:two_point2}) and equation (\ref{eq:cone0}), we obtain
	\begin{equation*}
	\|\mathbf{Z}\Delta\|_{NT}^2 \leq 2c^{-1}\lambda\left\{|\hat\alpha - \alpha|_1/\sqrt{N} + \Omega(\hat\beta - \beta) \right\} + 2\lambda\left\{\Omega_0(\hat\beta - \beta) - \Omega_1(\hat\beta - \beta)\right\}.
	\end{equation*}
	Since the left side of this equation is $\geq 0$, this shows that
	\begin{equation*}
	(1-c^{-1})\Omega_1(\hat\beta - \beta) \leq (1+c^{-1})\Omega_0(\hat\beta - \beta) + c^{-1}|\hat\alpha - \alpha|_1/\sqrt{N}
	\end{equation*}
	or equivalently
	\begin{equation}\label{eq:cone_condition_1}
	\Omega_1(\hat\beta - \beta) \leq \frac{c+1}{c-1}\Omega_0(\hat\beta - \beta) + (c-1)^{-1}|\hat\alpha - \alpha|_1/\sqrt{N}.
	\end{equation}
	Put $\Delta_N = ((\hat\alpha - \alpha)^\top/\sqrt{N},(\hat\beta - \beta)^\top)^\top$. Then under Assumption~\ref{as:covariance}
	\begin{equation*}
	\begin{aligned}
	|\Delta_N|_1 & \lesssim \Omega(\hat\beta - \beta) + |\hat\alpha - \alpha|_1/\sqrt{N} \\
	& \leq \frac{2c}{c-1}\Omega_0(\hat\beta - \beta) + \frac{c}{c-1}|\hat\alpha - \alpha|_1/\sqrt{N}  \\
	& \lesssim |\hat\alpha - \alpha|_2 + \sqrt{s}|\hat\beta - \beta|_2 \\
	& \leq \sqrt{s\vee N|\Delta_N|^2_2} \\
	& \lesssim \sqrt{s\vee N|\Sigma^{1/2}\Delta_N|^2_2} \\
	& = \sqrt{s\vee N\left\{\|\mathbf{Z}\Delta\|_{NT}^2 + \Delta_N^\top(\hat\Sigma - \Sigma)\Delta_N \right\} }\\
	& \leq \sqrt{s\vee N\left\{\|\mathbf{Z}\Delta\|_{NT}^2 + |\Delta_N|_1^2|\mathrm{vech}(\hat\Sigma - \Sigma)|_\infty \right\}} \\
	& \lesssim \sqrt{s\vee N\left\{\lambda|\Delta_N|_1 + |\Delta_N|_1^2|\mathrm{vech}(\hat\Sigma - \Sigma)|_\infty \right\}}.
	\end{aligned}
	\end{equation*}
	Consider the following event $E = \{|\mathrm{vech}(\hat\Sigma - \Sigma)|_\infty<1/(2s\vee N)\}$. Under Assumption~\ref{as:data} by Theorem~\ref{thm:fn_long} and \cite{babii2020inference}, Theorem 3.1
	\begin{equation*}
	\begin{aligned}
	\Pr(E^c) & \leq \Pr\left(\max_{i\in[N],j\in[p]}\left|\frac{1}{\sqrt{N}T}\sum_{t=1}^T\left\{x_{i,t,j} - \E[x_{i,t,j}]\right\}\right|\geq \frac{1}{2s\vee N}\right) \\
	& \qquad + \Pr\left(\max_{1\leq j\leq k\leq p}\left|\frac{1}{NT}\sum_{i=1}^N\sum_{t=1}^Tx_{i,t,j}x_{i,t,k} - \E[x_{i,t,j}x_{i,t,k}]\right|\geq \frac{1}{2s\vee N}\right) \\
	& \lesssim p(s\vee N)^{\tilde \kappa} T^{1-\tilde\kappa}(N^{1-\tilde\kappa/2} + pN^{1-\tilde\kappa}) + p(p\vee N)e^{-cNT/(s\vee N)^2}.
	\end{aligned}
	\end{equation*}	
	Therefore, on the event $E$
	\begin{equation*}
	|\hat\alpha - \alpha|_1/\sqrt{N} + |\hat\beta -\beta|_1 = |\Delta_N|_1\lesssim (s\vee N)\lambda,
	\end{equation*}
	and whence from equation (\ref{eq:two_point2}) we obtain
	\begin{equation*}
	\begin{aligned}
	\|\mathbf{Z}\Delta\|_{NT}^2 & \lesssim \lambda\left\{|\hat\alpha - \alpha|_1/\sqrt{N} + \Omega(\hat\beta - \beta) \right\} \\
	& \lesssim \lambda|\Delta_N|_1  \leq (s\vee N)\lambda^2.
	\end{aligned}
	\end{equation*}
Suppose now that $\|\mathbf{m} - \mathbf{Z}\rho\|_{NT} > \frac{1}{2}\|\mathbf{Z}\Delta\|_{NT}$. Then, obviously,
	\begin{equation*}
	\|\mathbf{Z}(\hat\rho - \rho)\|^2_{NT}  \leq 4\|\mathbf{m} - \mathbf{Z}\rho\|^2_{NT}.
	\end{equation*}
	Therefore, on the event $E$, we always have
	\begin{equation*}
	\|\mathbf{Z}(\hat\rho - \rho)\|^2_{NT}  \lesssim (s\vee N)\lambda^2 +  4\|\mathbf{m} - \mathbf{Z}\rho\|^2_{NT},
	\end{equation*}
	which proves the statement of the theorem.
\end{proof}

\begin{proof}[Proof of Theorem~\ref{thm:pooled_clt}]
	By Fermat's rule, the pooled sg-LASSO estimator in equation~(\ref{eq:pooled_panel}) satisfies
	\begin{equation*}
	\mathbf{Z}^\top(\mathbf{Z}\hat\rho - \mathbf{y})/NT + \lambda z^* = 0
	\end{equation*}
	for some $z^*\in\partial\Omega(\hat\rho)$. Rearranging this expression and multiplying by $\hat\Theta$
	\begin{equation*}
	\hat\rho - \rho + \hat\Theta\lambda z^* = \hat\Theta\mathbf{Z}^\top\mathbf{u}/NT + (I-\hat\Theta\hat\Sigma)(\hat\rho - \rho) + \hat\Theta\mathbf{Z}^\top(\mathbf{m} - \mathbf{Z}\rho)/NT,
	\end{equation*}
	where we use $\hat\Sigma = \mathbf{Z}^\top\mathbf{Z}/NT$ and $\mathbf{y}=\mathbf{m}+\mathbf{u}$. Plugging $\lambda z^*$ from the first-order conditions and multiplying by $\sqrt{NT}$
	\begin{equation*}
	\begin{aligned}
	\sqrt{NT}(\hat\rho - \rho + B) = \hat\Theta\mathbf{Z}^\top\mathbf{u}/\sqrt{NT} + \sqrt{NT}(I-\hat\Theta\hat\Sigma)(\hat\rho - \rho) + \hat\Theta\mathbf{Z}^\top(\mathbf{m} - \mathbf{Z}\rho)/\sqrt{NT}.
	\end{aligned}
	\end{equation*}
	Then for a group of regression coefficients $G\subset[p+1]$, we have
	\begin{equation*}
	\begin{aligned}
	\sqrt{NT}(\hat\rho_G - \rho_G + B_G) & = \frac{1}{\sqrt{NT}}\sum_{i=1}^N\sum_{t=1}^Tu_{i,t}\Theta_Gz_{i,t} + \frac{1}{\sqrt{NT}}\sum_{i=1}^N\sum_{t=1}^Tu_{i,t}(\hat\Theta_G - \Theta_G)z_{i,t} \\
	& \qquad + \sqrt{NT}(I - \hat\Theta\hat\Sigma)_G(\hat\rho - \rho) + \hat\Theta_G\mathbf{Z}^\top(\mathbf{m} - \mathbf{Z}\rho)/\sqrt{NT} \\
	& \triangleq I_{N,T} + II_{N,T} + III_{N,T} + IV_{N,T}.
	\end{aligned}
	\end{equation*}
	We will show that by Theorem~\ref{thm:clt}, $I_{N,T}\xrightarrow{d}N(0,\Xi_G)$ as $N,T\to\infty$. To that end, by Minkowski's inequality under Assumptions~\ref{as:data} (i) and \ref{as:clt} (ii)
	\begin{equation*}
	\begin{aligned}
	\max_{i\in[N],j\in G}\|u_{i,t}\Theta_j z_{i,t}\|_q & \leq \max_{i\in[N],j\in G}\sum_{k=1}^{p+1}\|u_{i,t}z_{i,t,k}\Theta_{j,k}\|_q \\
	& \leq \|\Theta_G\|_\infty\max_{i\in[N],j\in G,k\in[p+1]}\|u_{i,t}z_{i,t,k}\|_q = O(1).
	\end{aligned}
	\end{equation*}	
	Lastly, under Assumption~\ref{as:clt} (i), for every $i,N\in\Nn$,
	\begin{equation*}
	\begin{aligned}
	\lim_{T\to\infty}\Var(u_{i,t}\Theta_Gz_{i,t}) & = \lim_{T\to\infty}\Theta_G\Var(u_{i,t}z_{i,t})\Theta_G^\top \\
	& \lesssim \lim_{T\to\infty} \Theta_G\Sigma\Theta_G  = (\Theta_G^\top)_G <\infty
	\end{aligned}
	\end{equation*} 
	since groups have a fixed size. In conjunction with Assumption~\ref{as:data} (ii), this verifies conditions of Theorem~\ref{thm:clt} and shows that $I_{N,T}\xrightarrow{d}N(0,\Xi_G)$.
	
	Next,
	\begin{equation*}
	\begin{aligned}
	|II_{N,T}| & \leq \|\hat\Theta_G - \Theta_G\|_\infty\left|\frac{1}{\sqrt{NT}}\sum_{i=1}^N\sum_{t=1}^Tu_{i,t}z_{i,t}\right|_\infty \\
	& = O_P\left(\frac{Sp^{1/\kappa}}{(NT)^{1-1/\kappa}}\vee S\sqrt{\frac{\log p}{NT}}\right)O_P\left(\frac{p^{1/\kappa}}{(NT)^{1/2-1/\kappa}}\vee\sqrt{\log p}\right)  = o_P(1),
	\end{aligned}
	\end{equation*}
	where we use Proposition~\ref{prop:precision} and Theorem~\ref{thm:fn_long}. Similarly by Proposition~\ref{prop:precision} and Corollary~\ref{cor:pooled}
	\begin{equation*}
	\begin{aligned}
	|III_{N,T}| & \leq \sqrt{NT}\max_{j\in G}|(I - \hat\Theta\hat\Sigma)_j|_\infty|\hat\rho - \rho|_1 \\
	& = O_P\left(\frac{p^{1/\kappa}}{(NT)^{1/2-1/\kappa}}\vee\sqrt{\log p}\right)O_P\left(\frac{sp^{1/\kappa}}{(NT)^{1-1/\kappa}}\vee s\sqrt{\frac{\log p}{NT}}\right)  = o_P(1).
	\end{aligned}
	\end{equation*}
	Lastly, by the Cauchy-Schwartz inequality
	\begin{equation*}
	\begin{aligned}
	|IV_{N,T}|_\infty & \leq \max_{j\in G}|\mathbf{Z}\hat\Theta_j^\top|_2\|\mathbf{m} - \mathbf{Z}\rho\|_{NT}  = \max_{j\in G}\sqrt{\hat\Theta_j^\top\hat\Sigma\hat\Theta_j}o_P(1) \\
	& \leq \|\hat\Theta_G\|_\infty\sqrt{|\mathrm{vech}(\hat\Sigma)|_\infty} o_P(1)  = o_P(1),
	\end{aligned}
	\end{equation*}
	where the second line follows under Assumption~\ref{as:clt} (v), and the last by Proposition~\ref{prop:precision} and Theorem~\ref{thm:fn_long} under maintained assumptions.
\end{proof}

\begin{proposition}\label{prop:precision}
	Suppose that Assumptions~\ref{as:data}, \ref{as:covariance}, \ref{as:tuning}, \ref{as:rates}, and \ref{as:clt} are satisfied for each nodewise regression $j\in G$. Then if $S^\kappa p (NT)^{1-\kappa} \to 0$ and $S^2\log p/NT$ $\to$ 0
	\begin{equation*}
		\|\hat\Theta_G - \Theta_G\|_\infty = O_P\left(\frac{Sp^{1/\kappa}}{(NT)^{1-1/\kappa}}\vee S\sqrt{\frac{\log p}{NT}}\right)
	\end{equation*}
	and
	\begin{equation*}
		\max_{j\in G}|(I - \hat\Theta\hat\Sigma)_j|_\infty = O_P\left(\frac{p^{1/\kappa}}{(NT)^{1-1/\kappa}}\vee \sqrt{\frac{\log p}{NT}}\right).
	\end{equation*}
\end{proposition}
\begin{proof}
	The proof is similar to the proof of \cite{babii2020inference}, Propositions A.1.2 and A.1.3.
\end{proof}

\section{Concentration and moment inequalities}\label{sec:concentration}
In this section we present a suitable for us Rosenthal's moment inequality for dependent data and a new Fuk-Nagaev concentration inequality for panel data reflecting the concentration jointly over $N$ and $T$. 

For a random vector $\xi_{i,t}=(\xi_{i,t,1},\dots,\xi_{i,t,p})\in\R^p$, let $\tau_k^{(i,j)}$ denote the $\tau$-mixing coefficient of $\xi_{i,t,j}$. The following result describes a Fuk-Nagaev concentration inequality for panel data. It is worth mentioning that the inequality does not follow from \cite{babii2020inference} and is of independent interest for the high-dimensional panel data.\footnote{The direct application of the time series Fuk-Nagaev inequality of \cite{babii2020inference} leads to inferior concentration results for panel data.}
\begin{theorem}\label{thm:fn_long}
	Let $\{\xi_{i,t}: i\in[N],t\in[T]\}$ be an array of centered random vectors in $\R^p$ such that $(\xi_{i,1},\dots,\xi_{i,T})$ are independent over $i$ and (i) $\max_{i\in[N],t\in[T],j\in[p]}\|\xi_{i,t,j}\|_q=O(1)$ for some $q>2$; (ii) $\max_{i\in[N],j\in[p]}\tau_{k}^{(i,j)}= O(k^{-a})$ for some $a>(q-1)/(q-2)$. Then for every $u>0$
	\begin{equation*}
	\Pr\left(\left|\sum_{i=1}^N\sum_{t=1}^T\xi_{i,t}\right|_\infty > u \right) \leq c_1pNTu^{-\kappa} + 4pe^{-c_2u^2/NT}
	\end{equation*}
	for some universal constants $c_1,c_2>0$ and $\kappa = ((a+1)q - 1)/(a+q-1)$.
\end{theorem}
\begin{proof}[Proof of Theorem~\ref{thm:fn_long}]
	Suppose first that $p=1$. For $a\in\R$ with some abuse of notation, let $[[a]]$ denote its integer part. For each $i\in[N]$, split the partial sum into blocks with at most $J\in\Nn$ summands
	\begin{equation*}
	\begin{aligned}
	V_{i,k} & = \xi_{i,(k-1)J+1} + \dots + \xi_{i,kJ},\qquad k=1,2,\dots,[[T/J]] \\
	V_{i,[[T/J]]+1} & = \xi_{i,[[T/J]]J+1} + \dots + \xi_{i,T},
	\end{aligned}
	\end{equation*}
	where we set $V_{i,[[T/J]]+1}=0$ if $[[T/J]]J=T$. Let $\{U_{i,t}:i\in[N],t\in[T]\}$ be i.i.d.\ random variables uniformly distributed on $(0,1)$ and independent of $\{\xi_{i,t}:i\in[N],t\in[T]\}$. Put $\mathcal{M}_{i,t}=\sigma(V_{i,1},\dots,V_{i,t-2})$ for every $t\geq 3$. For each $i\in[N]$, if $t=1,2$, set $V_{i,t}^* = V_{i,t}$, while if $t\geq 3$, then by \cite{dedecker2004coupling}, Lemma 5, there exist random variables $V_{i,t}^*=_dV_{i,t}$ such that
	\begin{enumerate}
		\item $V_{i,t}^*$ is $\mathcal{M}_{i,t}\vee\sigma(V_{i,t})\vee \sigma(U_{i,t})$-measurable.
		\item $V_{i,t}^*\si(V_{i,1},\dots,V_{i,t-2})$.
		\item $\|V_{i,t} - V_{i,t}^*\|_1 = \tau(\mathcal{M}_{i,t},V_{i,t})$.
	\end{enumerate}
	Property 1.\ implies that there exists a measurable function $f_i$ such that
	\begin{equation*}
	V_{i,t}^* = f_i(V_{i,t},V_{i,t-2},\dots,V_{i,1},U_{i,t}).
	\end{equation*}
	Property 2.\ implies that $(V_{i,2t}^*)_{t\geq 1}$ and $(V_{i,2t-1}^*)_{t\geq 1}$ are sequences of independent random variables for every $i\in[N]$. Moreover, $\{V_{i,2t}^*:i\in[N],t\geq 1 \}$ and $\{V_{i,2t-1}^*:i\in[N],t\geq 1 \}$ are sequences of independent random variables since $\{\xi_{i,t}:\;t\in[T]\}$ are independent over $i\in[N]$. 
	
	Decompose
	\begin{equation*}
	\begin{aligned}
	\left|\sum_{i=1}^N\sum_{t=1}^T\xi_{i,t}\right| & \leq \left|\sum_{i=1}^N\sum_{t\geq 1}V_{i,2t}^*\right| + \left|\sum_{i=1}^N\sum_{t\geq 1}V_{i,2t-1}^*\right| + \sum_{i=1}^N\sum_{t=3}^{[[T/J]]+1}\left|V_{i,t} - V_{i,t}^*\right| \\
	& \triangleq I + II + III.
	\end{aligned}
	\end{equation*}
	By \cite{fuk1971probability}, Corollary 4 for independent data there exist constants $c_1,c_2>0$ such that
	\begin{equation*}
	\begin{aligned}
	\Pr(I> u/3) & \leq c_1u^{-q}\sum_{i=1}^N\sum_{t\geq 1}\E|V_{i,2t}^*|^q + 2\exp\left(-\frac{c_2u^2}{\sum_{i=1}^N\sum_{t\geq 1}\Var(V_{i,2t}^*)}\right) \\
	& \leq c_1u^{-q}\sum_{i=1}^N\sum_{t\geq 1}\E|V_{i,2t}|^q + 2\exp\left(-\frac{c_2u^2}{NT}\right),
	\end{aligned}
	\end{equation*}
	where we use $V_{i,t}^* =_d V_{i,t}$ and
$	\sum_{i=1}^N\sum_{t\geq 1}\Var(V_{i,2t})$ = $O(T),$
	which follows from \cite{babii2020inference}, Lemma A.1.2 under assumptions (i) and (ii). Similarly,
	\begin{equation*}
	\Pr(II> u/3) \leq c_1u^{-q}\sum_{i=1}^N\sum_{t\geq 1}\E|V_{i,2t}|^q + 2\exp\left(-\frac{c_2u^2}{NT}\right).
	\end{equation*}
	Finally, since $\mathcal{M}_{i,t}$ and $V_{i,t}$ are separated by $J+1$ lags of $\xi_{i,t}$, we have $\tau(\mathcal{M}_{i,t},V_{i,t}) \leq J\tau_J^{(i,j)}(J+1)$. By Markov's inequality and property 3., this gives
	\begin{equation*}
	\begin{aligned}
	\Pr(III>u/3) & \leq \frac{3}{u}\sum_{i=1}^N\sum_{t=3}^{[[T/J]]+1}\|V_{i,t} - V_{i,t}^*\|_1  \leq \frac{3NT}{u}\max_{i\in[N]}\tau_{J+1}^{(i,1)}.
	\end{aligned}
	\end{equation*}
	Combining all estimates together under (i)-(ii)
	\begin{equation*}
	\begin{aligned}
	\Pr\left(\left|\sum_{i=1}^N\sum_{t=1}^T\xi_{i,t}\right| > u \right) & \leq \Pr(I > u/3) + \Pr(II > u/3) + \Pr(III > u/3) \\
	& \leq c_1u^{-q}N\sum_{i=1}^N\sum_{t\geq 1}\|V_{i,t}\|^q_q + 4e^{-c_2u^2/NT} + \frac{3NT}{u}\max_{i\in[N]}\tau_{J+1}^{(i,1)} \\
	& \leq c_1u^{-q}J^{q-1}NT + \frac{3NT}{u}(J+1)^{-a} + 4e^{-c_2u^2/NT}
	\end{aligned}
	\end{equation*}
	for some constants $c_1,c_2>0$. To balance the first two terms, we shall choose the length of blocks $J\sim u^{\frac{q-1}{q+a-1}}$, in which case we get
	\begin{equation*}
	\Pr\left(\left|\sum_{i=1}^N\sum_{t=1}^T\xi_{i,t}\right| > u \right) \leq c_1NTu^{-\kappa} + 4e^{-c_2u^2/NT}
	\end{equation*}
	for some $c_1,c_2>0$. Finally, for $p>1$, the result follows by the union bound.
\end{proof}
It follows from Theorem~\ref{thm:fn_long} that there exists $C>0$ such that for every $\delta\in(0,1)$
\begin{equation*}
\Pr\left(\left|\frac{1}{NT}\sum_{t=1}^T\sum_{i=1}^N\xi_{i,t}\right|_\infty \leq C\left(\frac{p}{\delta (NT)^{\kappa - 1}}\right)^{1/\kappa}\vee \sqrt{\frac{\log(p/\delta)}{NT}} \right)\geq 1 - \delta.
\end{equation*}
Note that the inequality reflects the concentration jointly over $N$ and $T$ and that tails and persistence play an important role through the mixing-tails exponent $\kappa$. The inequality is a key technical tool that allows us to handle panel data with heavier than Gaussian tails and non-negligible $T$ and $N$. It is worth mentioning that the concentration over $N$ is also influenced by the weak dependence, which probably can be relaxed with a sharper proof technique. However, for geometrically ergodic processes, e.g., for stationary $AR(p)$, we have $\kappa\approx q$, in which case the time series dependence does not influence the concentration at all.

Let $(\xi_t)_{t\in\Nn}$ be a real-valued stochastic process, and let $Q_t$ denote the generalized inverse of the tail function $x\mapsto\Pr(|\xi_t|\geq x)$. Let $\xi\in\R$ be a random variable corresponding to $(\xi_t)_{t\in\Z}$ such that $Q\geq \sup_{t\in\Nn}Q_t$, where $Q$ is a generalized inverse of $x\mapsto\Pr(|\xi|\geq x)$. The following Rosenthal's moment inequality for $\tau$-dependent sequences follows from \cite{dedecker2004coupling}; see also \cite{dedecker2003new}.
\begin{theorem}\label{thm:rosenthal}
	Let $(\xi_t)_{t\in\Nn}$ be a centered stochastic process such that (i) there exists $q>2$ such that $\|\xi\|_q<\infty$, where $\xi\in\R$ corresponds to $(\xi_t)_{t\in\Nn}$; (ii) the $\tau$-mixing coefficients are $\tau_{k-1}\leq ck^{-a},\forall k\geq 1$ for some universal constants $c>0$ and $a>(q(r-2)+1)/(q-r)$. Then for every $r\in[2,q)$
	\begin{equation*}
	\E\left|\sum_{t=1}^T\xi_t\right|^r \leq c_{q,r}\left(T^{r/2}\|\xi\|_q^{qr/2(q-1)} + T\|\xi\|_q^{q(r-1)/(q-1)} \right),
	\end{equation*}
	where the constant $c_{q,r}$ depends only on $q$ and $r$.
\end{theorem}
\begin{proof}
	Let $G$ be the inverse of $x\mapsto \int_0^xQ(u)\dx u$ and put $H(u)=\sum_{k=0}^\infty\one_{2u<\tau_k}$, where $(\tau_k)_{k\in\Nn}$ are $\tau$-mixing coefficients of $(\xi_t)_{t\in\Nn}$. Note that for every $q\geq 1$,
	\begin{equation*}
	\int_0^{\|\xi\|_1}|Q\circ G(u)|^{q-1}\dx u = \int_0^1Q^q(v)\dx v = \|\xi\|_q^q.
	\end{equation*}
	Then by H\"{o}lder's inequality
	\begin{equation*}
	\begin{aligned}
	\int_0^{\|\xi\|_1}|H(u)Q\circ G(u)|^{r-1}\dx u & \leq \left(\int_0^{\|\xi\|_1}H^{(q-1)(r-1)/(q-r)}(u)\dx u\right)^\frac{q-1}{q-r}\|\xi\|_q^{q(r-1)/(q-1)} \\
	\end{aligned}
	\end{equation*}
	Note also that for some constant $C_{q,r}$ that depends only on $q$ and $r$ we have
	\begin{equation*}
	\begin{aligned}
	\int_0^{\|\xi\|_1}H^{(q-1)(r-1)/(q-r)}(u)\dx u & \leq (1\vee s_{q,r})\int_0^{\|\xi\|_1}\sum_{k=0}^\infty(k+1)^{(q-1)(r-1)/(q-r)-1}\one_{2u<\tau_k}\dx u \\
	& \leq 0.5(1\vee s_{q,r})\sum_{k=0}^\infty(k+1)^{(q-1)(r-1)/(q-r)-1}\tau_k \\
	& \leq 0.5c(1\vee s_{q,r})\sum_{k=1}^\infty k^{(q-1)(r-1)/(q-r)-1-a} \\
	& \leq C_{q,r}
	\end{aligned}
	\end{equation*}
	where we use the fact that $H^s(u)=\sum_{k=0}^\infty((k+1)^s - k^s)\one_{2u<\tau_k}$, $(k+1)^s - k^s \leq (1\vee s)(k+1)^{s-1}$ with $s=s_{q,r}=(q-1)(r-1)/(q-r)$, and the series converges since $a>(q(r-2)+1)/(q-r)$. Combining these estimates
	\begin{equation}\label{eq:quantile_inequality}
	\int_0^{\|\xi\|_1}|H(u)Q\circ G(u)|^{r-1}\dx u \leq C_{q,r}^\frac{q-1}{q-r}\|\xi\|_q^{q(r-1)/(q-1)}.
	\end{equation}
	By \cite{dedecker2004coupling}, Corollary 1, for some constant $c_r>0$ that depends only on $r$
	\begin{equation*}
	\begin{aligned}
	\E\left|\sum_{t=1}^T\xi_t\right|^r & \leq c_r\left\{\left(T\int_0^{\|\xi\|_1}H(u)Q\circ G(u)\dx u\right)^{r/2} + T\int_0^{\|\xi\|_1}|H(u)Q\circ G(u)|^{r-1}\dx u\right\} \\
	& \leq c_r\left\{T^{r/2}\left(C_{q,r}^\frac{q-1}{q-2}\|\xi\|_q^{q/(q-1)}\right)^{r/2} + TC_{q,r}^\frac{q-1}{q-r}\|\xi\|_q^{q(r-1)/(q-1)} \right\} \\
	& \leq c_{q,r}\left(T^{r/2}\|\xi\|_q^{qr/2(q-1)} + T\|\xi\|_q^{q(r-1)/(q-1)} \right),
	\end{aligned}
	\end{equation*}
	where the second line follows by equation~(\ref{eq:quantile_inequality}) and $c_{q,r}>0$ depends only on $q$ and $r$.
\end{proof}

\section{Large $N$ and $T$ central limit theorem}
For a double sequence $\{a_{N,T}:N,T\in\Nn\}$, we use $\lim_{N,T\to\infty}a_{N,T}$ to denote the limit when $N,T\to\infty$ jointly and $\max_{N,T\in\Nn}a_{N,T}=\max\{a_{N,T}:N\in\Nn,T\in\Nn\}$. The following central limit theorem holds for panel data consisting of $\tau$-mixing processes that may change over $N$ and $T$.
\begin{theorem}\label{thm:clt}
	Let $\{\xi_{N,T,i,t}:i\in\Nn,t\in\Z \}$ be an array of centered random vectors in $\R^p$ such that for each $N,T$, and $i$, $\{\xi_{N,T,i,t} :t\in\Z \}$ is a stationary process in $\R^p$ and $\{(\xi_{N,T,i,1},\dots,\xi_{N,T,i,T}):i\in\Nn\}$ are independent arrays in $\R^{p}\times\R^T$ satisfying (i) for some $q>2$, $\max_{i\in[N],j\in[p]}\|\xi_{N,T,i,t,j}\|_q=O(1)$; (ii) for all $N,T,i,j$, the $\tau$-mixing coefficients of $\{\xi_{N,T,i,t,j}:t\in\Z\}$ satisfy $\tau_{k-1}\leq ck^{-a},\forall k\geq 1$ for some universal constants $c>0$ and $a>(q-1)/(q-2)\vee (q\delta+1)/(q-2-\delta)$ with $q>2+\delta$ and $\delta>0$; (iii) for every $i,N\in\Nn$, $\lim_{T\to\infty}\Var(\xi_{N,T,i,t})<\infty$. Then
	\begin{equation*}
	\frac{1}{\sqrt{NT}}\sum_{i=1}^N\sum_{t=1}^T\xi_{N,T,i,t} \xrightarrow{\mathrm{d}} N(0,\Xi)\qquad \text{as}\qquad N,T\to\infty,
	\end{equation*}
	where $\Xi = \lim_{N,T\to\infty}\frac{1}{N}\sum_{i=1}^N\Var\left(\frac{1}{\sqrt{T}}\sum_{t=1}^T\xi_{N,T,i,t}\right)$ is a finite matrix, assumed to be a positive definite.
\end{theorem}
\begin{proof}
	By the Cram\'{e}r-Wold device, see \cite{billingsley1995probability}, Theorem 29.4,
	\begin{equation*}
	\frac{1}{\sqrt{NT}}\sum_{i=1}^N\sum_{t=1}^T\xi_{N,T,i,t} \xrightarrow{\mathrm{d}} N(0,\Xi)\qquad \text{as}\qquad N,T\to\infty
	\end{equation*}
	in $\R^p$ if and only if for every $z\in\R^p$, the following weak convergence holds in $\R$
	\begin{equation*}
	z^\top\left(\frac{1}{\sqrt{NT}}\sum_{i=1}^N\sum_{t=1}^T\xi_{N,T,i,t}\right) \xrightarrow{\mathrm{d}} N(0,z^\top\Xi z)\qquad \text{as}\qquad N,T\to\infty.
	\end{equation*}
	Note that under maintained assumptions, for each $N,T$ and $z\in\R^p$,
	\begin{equation*}
	z^\top\left(\frac{1}{\sqrt{NT}}\sum_{i=1}^N\sum_{t=1}^T\xi_{N,T,i,t}\right) = \sum_{i=1}^Nz^\top\left(\frac{1}{\sqrt{NT}}\sum_{t=1}^T\xi_{N,T,i,t}\right)
	\end{equation*}
	is a sum of $N$ independent zero-mean random variables. By independence and stationarity, the variance of this sum is
	\begin{equation*}
	\begin{aligned}
		\sigma_{N,T,z}^2 & \triangleq \frac{1}{N}\sum_{i=1}^N\Var\left(\frac{1}{\sqrt{T}}\sum_{t=1}^Tz^\top\xi_{N,T,i,t}\right) \\
		& = \frac{1}{N}\sum_{i=1}^N\left\{\Var(z^\top\xi_{N,T,i,t}) + 2\sum_{k=1}^{T-1}\left(1-\frac{k}{T}\right)\Cov(z^\top\xi_{N,T,i,0},z^\top\xi_{N,T,i,k}) \right\}.
	\end{aligned}
	\end{equation*}
	If we show that the limit in the parentheses exists for every $i,N\in\Nn$, then the joint limit of $\sigma^2_{N,T,z}$ as $N,T\to\infty$ is the same as the sequential limit
	\begin{equation*}
		\lim_{N\to\infty}\lim_{T\to\infty}\frac{1}{N}\sum_{i=1}^N\left\{\Var(z^\top\xi_{N,T,i,t}) + 2\sum_{k=1}^{T-1}\left(1-\frac{k}{T}\right)\Cov(z^\top\xi_{N,T,i,0},z^\top\xi_{N,T,i,k}) \right\};
	\end{equation*}
	see \cite{apostol1974mathematical}, Theorem 8.39. By \cite{babii2020inference}, Lemma A.1.1, for every $k\geq 1$
	\begin{equation*}
		|\Cov(z^\top\xi_{N,T,i,0},z^\top\xi_{N,T,i,k})| \leq \tau_k^\frac{q-2}{q-1}\|z^\top\xi_{N,T,i,0}\|_q^{q/(q-1)} = O(k^{-a}),
	\end{equation*}
	where the second inequality follows under (i)-(ii). Moreover, $\sum_{k=1}^\infty k^{-a}<\infty$ under (ii). Therefore, by Lebesgue's dominated convergence theorem, for every $i,N\in\Nn$,
	\begin{equation*}
		\lim_{T\to\infty}\sum_{k=1}^{T-1}\left(1-\frac{k}{T}\right)\Cov(z^\top\xi_{N,T,i,0},z^\top\xi_{N,T,i,k})<\infty,
	\end{equation*}
	and whence under (ii)
	\begin{equation*}
		\lim_{N,T\to\infty}\sigma^2_{N,T} = \lim_{N,T\to\infty}\frac{1}{N}\sum_{i=1}^N\Var\left(\frac{1}{\sqrt{T}}\sum_{t=1}^Tz^\top\xi_{N,T,i,t}\right) = z^\top\Xi z<\infty.
	\end{equation*}
	The statement of the theorem follows by the central limit theorem for independent random variables, provided that the following Lyapunov condition holds
	\begin{equation*}
	\lim_{N,T\to\infty}\frac{1}{(NT)^{1+\delta/2}}\sum_{i=1}^N\E\left| \sum_{t=1}^Tz^\top\xi_{N,T,i,t}\right|^{2+\delta}=0;
	\end{equation*}
	see \cite{billingsley1995probability}, Theorem 27.3 and \cite{phillips1999linear}, Theorem 2.
	
	By Theorem~\ref{thm:rosenthal}, for some $c_{q,\delta}$ that depends only on $q$ and $\delta$, 
	\begin{equation*}
	\begin{aligned}
		\E\left| \sum_{t=1}^Tz^\top\xi_{N,T,i,t}\right|^{2+\delta} & \leq c_{q,\delta}\left\{T^{1+\delta/2}\|z^\top\xi_{N,T,i,t}\|_q^{q(1+\delta/2)/(q-1)} + T\|z^\top\xi_{N,T,i,t}\|_q^{q(1+\delta)/(q-1)} \right\}.
	\end{aligned}
	\end{equation*}
	 Therefore, the Lyapunov condition holds under (i).
\end{proof}

\vfill
}

\newpage

\setcounter{page}{1}
\setcounter{section}{0}
\setcounter{equation}{0}
\setcounter{table}{0}
\setcounter{figure}{0}
\renewcommand{\theequation}{OA.\arabic{equation}}
\renewcommand\thetable{OA.\arabic{table}}
\renewcommand\thefigure{OA.\arabic{figure}}
\renewcommand\thesection{OA.\arabic{section}}
\renewcommand\thepage{Online Appendix - \arabic{page}}
\renewcommand\thetheorem{OA.\arabic{theorem}}

\begin{center}
	{\LARGE\textbf{ONLINE APPENDIX}}	
\end{center}

\section{Data description \label{appsec:data-description}}

\subsection{Firm-level data} 

The full list of firm-level data is provided in Table \ref{tab:data}. 
We also add two daily firm-specific stock market predictor variables: stock returns and a realized variance measure, which is defined as the rolling sample variance over the previous 60 days (i.e.\ 60-day historical volatility).

\subsubsection{Firm sample selection}

We select a sample of firms based on data availability. First, we remove all firms from I/B/E/S which have missing values in earnings time series. Next, we retain firms that we can match with CRSP dataset. Finally, we keep firms that we can match with the RavenPack dataset.

\subsubsection{Firm-specific text data}

We create a link table of RavenPack ID and PERMNO identifiers which enables us to merge I/B/E/S and CRSP data with firm-specific textual analysis generated data from RavenPack. The latter is a rich dataset that contains intra-daily news information about firms. There are several editions of the dataset; in our analysis, we use the Dow Jones (DJ) and Press Release (PR) editions. The former contains relevant information from Dow Jones Newswires, regional editions of the Wall Street Journal, Barron's and MarketWatch. The PR edition contains news data, obtained from various press releases and regulatory disclosures, on a daily basis from a variety of newswires and press release distribution networks, including exclusive content from PRNewswire, Canadian News Wire, Regulatory News Service, and others. The DJ edition sample starts at $1^{st}$ of January, 2000, and PR edition data starts at $17^{th}$ of January, 2004. 

\smallskip

We construct our news-based firm-level covariates by filtering only highly relevant news stories. More precisely, for each firm and each day, we filter out news that has the \emph{Relevance Score} (REL) larger or equal to 75, as is suggested by the RavenPack News Analytics guide and used by practitioners; see for example \cite{kolanovic2017big}. REL is a score between 0 and 100 which indicates how strongly a news story is linked with a particular firm. A score of zero means that the entity is vaguely mentioned in the news story, while 100 means the opposite. A score of 75 is regarded as a significantly relevant news story. After applying the REL filter, we apply a novelty of the news filter by using the \emph{Event Novelty Score} (ENS); we keep data entries that have a score of 100. Like REL, ENS is a score between 0 and 100. It indicates the novelty of a news story within a 24-hour time window. A score of 100 means that a news story was not already covered by earlier announced news, while a subsequently published news story score on a related event is discounted, and therefore its scores are less than 100. Therefore, with this filter, we consider only novel news stories. We focus on {\it five sentiment indices} that are available in both DJ and PR editions. They are as follows.

\paragraph{\bf Event Sentiment Score} (ESS), for a given firm, represents the strength of the news measured using surveys of financial expert ratings for firm-specific events. The score value ranges between 0 and 100 --- values above (below) 50 classify the news as being positive (negative), 50 being neutral. 

\paragraph{\bf Aggregate Event Sentiment} (AES) represents the ratio of positive events reported on a firm compared to the total count of events measured over a rolling 91-day window in a particular news edition (DJ or PR). An event with ESS $>$ 50 is counted as a positive entry while ESS $<$ 50 is negative. Neutral news (ESS = 50) and news that does not receive an ESS score do not enter into the AES computation. As ESS, the score values are between 0 and 100. 

\paragraph{\bf Aggregate Event Volume} (AEV) represents the count of events for a firm over the last 91 days within a certain edition. As in the AES case, news that receives a non-neutral ESS score is counted and therefore accumulates positive and negative news. 

\paragraph{\bf Composite Sentiment Score} (CSS) represents the news sentiment of a given news story by combining various sentiment analysis techniques. The direction of the score is determined by looking at emotionally charged words and phrases and by matching stories typically rated by experts as having short-term positive or negative share price impact. The strength of the scores is determined by intra-day price reactions modeled empirically using tick data from approximately 100 large-cap stocks. As for ESS and AES, the score takes values between 0 and 100, 50 being the neutral.  

\paragraph{\bf News Impact Projections} (NIP) represents the degree of impact a news flash has on the market over the following two-hour period. The algorithm produces scores to accurately predict a relative volatility --- defined as scaled volatility by the average of volatilities of large-cap firms used in the test set --- of each stock price measured within two hours following the news. Tick data are used to train the algorithm and produce scores, which take values between 0 and 100, 50 representing zero impact news.

\smallskip

For each firm and each day with firm-specific news, we compute the average value of the specific sentiment score. In this way, we aggregate across editions and groups, where the latter is defined as a collection of related news. We then map the indices that take values between 0 and 100 onto $[-1,1]$. Specifically, let \(x_i \in \{\text{ESS}, \text{AES}, \text{CSS}, \text{NIP}\}\) be the average score value for a particular day and firm. We map $x_i \mapsto \bar x_i\in[-1,1]$ by computing $\bar x_i$ = $(x_i - 50)/50.$ For days with no news, we impute zero values. Note that series are centered around zero, where zero value means zero impact news. Therefore, imputing zeros is the same as assuming that no news on a given day implies zero impact news, which we believe is a reasonable assumption.

\clearpage
{\scriptsize
	\begin{longtable}{r |ll cc}
		& Ticker & Firm name & PERMNO & RavenPack ID \\ 
		\hline
		1 & MMM & 3M & 22592 & 03B8CF \\ 
		2 & ABT & Abbott labs& 20482 & 520632 \\ 
		3 & AUD & Automatic data processing & 44644 & 66ECFD \\ 
		4 & ADTN & Adtran & 80791 & 9E98F2 \\ 
		5 & AEIS & Advanced energy industries & 82547 & 1D943E \\ 
		6 & AMG & Affiliated managers group & 85593 & 30E01D \\ 
		7 & AKST & A K steel holding & 80303 & 41588B \\ 
		8 & ATI & Allegheny technologies & 43123 & D1173F \\ 
		9 & AB & AllianceBernstein holding l.p. & 75278 & CB138D \\ 
		10 & ALL & Allstate corp. & 79323 & E1C16B \\ 
		11 & AMZN & Amazon.com & 84788 & 0157B1 \\ 
		12 & AMD & Advanced micro devices & 61241 & 69345C \\ 
		13 & DOX & Amdocs ltd. & 86144 & 45D153 \\ 
		14 & AMKR & Amkor technology & 86047 & 5C8D61 \\ 
		15 & APH & Amphenol corp. & 84769 & BB07E4 \\ 
		16 & AAPL & Apple & 14593 & D8442A \\ 
		17 & ADM & Archer daniels midland & 10516 & 2B7A40 \\ 
		18 & ARNC & Arconic & 24643 & EC821B \\ 
		19 & ATTA & AT\&T & 66093 & 251988 \\ 
		20 & AVY & Avery dennison corp. & 44601 & 662682 \\ 
		21 & BHI & Baker hughes & 75034 & 940C3D \\ 
		22 & BAC & Bank of america corp. & 59408 & 990AD0 \\ 
		23 & BAX & Baxter international inc. & 27887 & 1FAF22 \\ 
		24 & BBT & BB\&T corp. & 71563 & 1A3E1B \\ 
		25 & BDX & Becton dickinson \& co. & 39642 & 873DB9 \\ 
		26 & BBBY & Bed bath \& beyond inc. & 77659 & 9B71A7 \\ 
		27 & BHE & Benchmark electronics inc. & 76224 & 6CF43C \\ 
		28 & BA & Boeing co. & 19561 & 55438C \\ 
		29 & BK & Bank of new york mellon corp. & 49656 & EF5BED \\ 
		30 & BWA & BorgWarner inc. & 79545 & 1791E7 \\ 
		31 & BP & BP plc & 29890 & 2D469F \\ 
		32 & EAT & Brinker international inc. & 23297 & 732449 \\ 
		33 & BMY & Bristol-Myers squibb co. & 19393 & 94637C \\ 
		34 & BRKS & Brooks automation inc. & 81241 & FC01C0 \\ 
		35 & CA & CA technologies inc. & 25778 & 76DE40 \\ 
		36 & COG & Cabot oil \& gas corp. & 76082 & 388E00 \\ 
		37 & CDN & Cadence design systems inc. & 11403 & CC6FF5 \\ 
		38 & COF & Capital one financial corp. & 81055 & 055018 \\ 
		39 & CRR & Carbo ceramics inc. & 83366 & 8B66CE \\ 
		40 & CSL & Carlisle cos. & 27334 & 9548BB \\ 
		41 & CCL & Carnival corporation \& plc & 75154 & 067779 \\ 
		42 & CERN & Cerner corp. & 10909 & 9743E5 \\ 
		43 & CHRW & C.H. robinson worldwide inc. & 85459 & C659EB \\ 
		44 & SCHW & Charles schwab corp. & 75186 & D33D8C \\ 
		45 & CHKP & Check point software technologies ltd. & 83639 & 531EF1 \\ 
		46 & CHV & Chevron corp. & 14541 & D54E62 \\ 
		47 & CI & CIGNA corp. & 64186 & 86A1B9 \\ 
		48 & CTAS & Cintas corp. & 23660 & BFAEB4 \\ 
		49 & CLX & Clorox co. & 46578 & 719477 \\ 
		50 & KO & Coca-Cola co. & 11308 & EEA6B3 \\ 
		51 & CGNX & Cognex corp. & 75654 & 709AED \\ 
		52 & COLM & Columbia sportswear co. & 85863 & 5D0337 \\ 
		53 & CMA & Comerica inc. & 25081 & 8CF6DD \\ 
		54 & CRK & Comstock resources inc. & 11644 & 4D72C8 \\ 
		55 & CAG & ConAgra foods inc. & 56274 & FA40E2 \\ 
		56 & STZ & Constellation brands inc. & 69796 & 1D1B07 \\ 
		57 & CVG & Convergys corp. & 86305 & 914819 \\ 
		58 & COST & Costco wholesale corp. & 87055 & B8EF97 \\ 
		59 & CCI & Crown castle international corp. & 86339 & 275300 \\ 
		60 & DHR & Danaher corp. & 49680 & E124EB \\ 
		61 & DRI & Darden restaurants inc. & 81655 & 9BBFA5 \\ 
		62 & DVA & DaVita inc. & 82307 & EFD406 \\ 
		63 & DO & Diamond offshore drilling inc. & 82298 & 331BD2 \\ 
		64 & D & Dominion resources inc. & 64936 & 977A1E \\ 
		65 & DOV & Dover corp. & 25953 & 636639 \\ 
		66 & DOW & Dow chemical co. & 20626 & 523A06 \\ 
		67 & DHI & D.R. horton inc. & 77661 & 06EF42 \\ 
		68 & EMN & Eastman chemical co. & 80080 & D4070C \\ 
		69 & EBAY & eBay inc. & 86356 & 972356 \\ 
		70 & EOG & EOG resources inc. & 75825 & A43906 \\ 
		71 & EL & Estee lauder cos. inc. & 82642 & 14ED2B \\ 
		72 & ETH & Ethan allen interiors inc. & 79037 & 65CF8E \\ 
		73 & ETFC & E*TRADE financial corp. & 83862 & 28DEFA \\ 
		74 & XOM & Exxon mobil corp. & 11850 & E70531 \\ 
		75 & FII & Federated investors inc. & 86102 & 73C9E2 \\ 
		76 & FDX & FedEx corp. & 60628 & 6844D2 \\ 
		77 & FITB & Fifth third bancorp & 34746 & 8377DB \\ 
		78 & FISV & Fiserv inc. & 10696 & 190B91 \\ 
		79 & FLEX & Flex ltd. & 80329 & B4E00D \\ 
		80 & F & Ford motor co. & 25785 & A6213D \\ 
		81 & FWRD & Forward air corp. & 79841 & 10943B \\ 
		82 & BEN & Franklin resources inc. & 37584 & 5B6C11 \\ 
		83 & GE & General electric co. & 12060 & 1921DD \\ 
		84 & GIS & General mills inc. & 17144 & 9CA619 \\ 
		85 & GNTX & Gentex corp. & 38659 & CC339B \\ 
		86 & HAL & Halliburton Co. & 23819 & 2B49F4 \\ 
		87 & HLIT & Harmonic inc. & 81621 & DD9E41 \\ 
		88 & HIG & Hartford financial services group inc. & 82775 & 766047 \\ 
		89 & HAS & Hasbro inc. & 52978 & AA98ED \\ 
		90 & HLX & Helix energy solutions group inc. & 85168 & 6DD6BA \\ 
		91 & HP & Helmerich \& payne inc. & 32707 & 1DE526 \\ 
		92 & HSY & Hershey co. & 16600 & 9F03CF \\ 
		93 & HES & Hess corp. & 28484 & D0909F \\ 
		94 & HON & Honeywell international inc. & 10145 & FF6644 \\ 
		95 & JBHT & J.B. Hunt transport services Inc. & 42877 & 72DF04 \\ 
		96 & HBAN & Huntington bancshares inc. & 42906 & C9E107 \\ 
		97 & IBM & IBM corp. & 12490 & 8D4486 \\ 
		98 & IEX & IDEX corp. & 75591 & E8B21D \\ 
		99 & IR & Ingersoll-Rand plc & 12431 & 5A6336 \\ 
		100 & IDTI & Integrated device technology inc. & 44506 & 8A957F \\ 
		101 & INTC & Intel corp. & 59328 & 17EDA5 \\ 
		102 & IP & International paper co. & 21573 & 8E0E32 \\ 
		103 & IIN & ITT corp. & 12570 & 726EEA \\ 
		104 & JAKK & Jakks pacific inc. & 83520 & 5363A2 \\ 
		105 & JNJ & Johnson \& johnson & 22111 & A6828A \\ 
		106 & JPM & JPMorgan chase \& co. & 47896 & 619882 \\ 
		107 & K & Kellogg co. & 26825 & 9AF3DC \\ 
		108 & KMB & Kimberly-Clark corp. & 17750 & 3DE4D1 \\ 
		109 & KNGT & Knight transportation inc. & 80987 & ED9576 \\ 
		110 & LSTR & Landstar system inc. & 78981 & FD4E8D \\ 
		111 & LSCC & Lattice semiconductor corp. & 75854 & 8303CD \\ 
		112 & LLY & Eli lilly \& co. & 50876 & F30508 \\ 
		113 & LFUS & Littelfuse inc. & 77918 & D06755 \\ 
		114 & LNC & Lincoln national corp. & 49015 & 5C7601 \\ 
		115 & LMT & Lockheed martin corp. & 21178 & 96F126 \\ 
		116 & MTB & M\&T bank corp. & 35554 & D1AE3B \\ 
		117 & MANH & Manhattan associates inc. & 85992 & 031025 \\ 
		118 & MAN & ManpowerGroup inc. & 75285 & C0200F \\ 
		119 & MAR & Marriott international inc. & 85913 & 385DD4 \\ 
		120 & MMC & Marsh \& mcLennan cos. & 45751 & 9B5968 \\ 
		121 & MCD & McDonald's corp. & 43449 & 954E30 \\ 
		122 & MCK & McKesson corp. & 81061 & 4A5C8D \\ 
		123 & MDU & MDU resources group inc. & 23835 & 135B09 \\ 
		124 & MRK & Merck \& co. inc. & 22752 & 1EBF8D \\ 
		125 & MTOR & Meritor inc & 85349 & 00326E \\ 
		126 & MTG & MGIC investment corp. & 76804 & E28F22 \\ 
		127 & MGM & MGM resorts international & 11891 & 8E8E6E \\ 
		128 & MCHP & Microchip technology inc. & 78987 & CDFCC9 \\ 
		129 & MU & Micron technology inc. & 53613 & 49BBBC \\ 
		130 & MSFT & Microsoft corp. & 10107 & 228D42 \\ 
		131 & MOT & Motorola solutions inc. & 22779 & E49AA3 \\ 
		132 & MSM & MSC industrial direct co.& 82777 & 74E288 \\ 
		133 & MUR & Murphy oil corp. & 28345 & 949625 \\ 
		134 & NBR & Nabors industries ltd. & 29102 & E4E3B7 \\ 
		135 & NOI & National oilwell varco inc. & 84032 & 5D02B7 \\ 
		136 & NYT & New york times co. & 47466 & 875F41 \\ 
		137 & NFX & Newfield exploration co. & 79915 & 9C1A1F \\ 
		138 & NEM & Newmont mining corp. & 21207 & 911AB8 \\ 
		139 & NKE & NIKE inc. & 57665 & D64C6D \\ 
		140 & NBL & Noble energy inc. & 61815 & 704DAE \\ 
		141 & NOK & Nokia corp. & 87128 & C12ED9 \\ 
		142 & NOC & Northrop grumman corp. & 24766 & FC1B7B \\ 
		143 & NTRS & Northern trust corp. & 58246 & 3CCC90 \\ 
		144 & NUE & NuCor corp. & 34817 & 986AF6 \\ 
		145 & ODEP & Office depot inc. & 75573 & B66928 \\ 
		146 & ONB & Old national bancorp & 12068 & D8760C \\ 
		147 & OMC & Omnicom group inc. & 30681 & C8257F \\ 
		148 & OTEX & Open text corp. & 82833 & 34E891 \\ 
		149 & ORCL & Oracle corp. & 10104 & D6489C \\ 
		150 & ORBK & Orbotech ltd. & 78527 & 290820 \\ 
		151 & PCAR & Paccar inc. & 60506 & ACF77B \\ 
		152 & PRXL & Parexel international corp. & 82607 & EF8072 \\ 
		153 & PH & Parker hannifin corp. & 41355 & 6B5379 \\ 
		154 & PTEN & Patterson-uti energy inc. & 79857 & 57356F \\ 
		155 & PBCT & People's united financial inc. & 12073 & 449A26 \\ 
		156 & PEP & PepsiCo inc. & 13856 & 013528 \\ 
		157 & PFE & Pfizer inc. & 21936 & 267718 \\ 
		158 & PIR & Pier 1 imports inc. & 51692 & 170A6F \\ 
		159 & PXD & Pioneer natural resources co. & 75241 & 2920D5 \\ 
		160 & PNCF & PNC financial services group inc. & 60442 & 61B81B \\ 
		161 & POT & Potash corporation of saskatchewan inc. & 75844 & FFBF74 \\ 
		162 & PPG & PPG industries inc. & 22509 & 39FB23 \\ 
		163 & PX & Praxair inc. & 77768 & 285175 \\ 
		164 & PG & Procter \& gamble co. & 18163 & 2E61CC \\ 
		165 & PTC & PTC inc. & 75912 & D437C3 \\ 
		166 & PHM & PulteGroup inc. & 54148 & 7D5FD6 \\ 
		167 & QCOM & Qualcomm inc. & 77178 & CFF15D \\ 
		168 & DGX & Quest diagnostics inc. & 84373 & 5F9CE3 \\ 
		169 & RL & Ralph lauren corp. & 85072 & D69D42 \\ 
		170 & RTN & Raytheon co. & 24942 & 1981BF \\ 
		171 & RF & Regions financial corp. & 35044 & 73C521 \\ 
		172 & RCII & Rent-a-center inc. & 81222 & C4FBDC \\ 
		173 & RMD & ResMed inc. & 81736 & 434F38 \\ 
		174 & RHI & Robert half international inc. & 52230 & A4D173 \\ 
		175 & RDC & Rowan cos. inc. & 45495 & 3FFA00 \\ 
		176 & RCL & Royal caribbean cruises ltd. & 79145 & 751A74 \\ 
		177 & RPM & RPM international inc. & 65307 & F5D059 \\ 
		178 & RRD & RR R.R. donnelley \& sons co. & 38682 & 0BE0AE \\ 
		179 & SLB & Schlumberger ltd. n.v. & 14277 & 164D72 \\ 
		180 & SCTT & Scotts miracle-gro co. & 77300 & F3FCC3 \\ 
		181 & SM & SM st. mary land \& exploration co. & 78170 & 6A3C35 \\ 
		182 & SONC & Sonic corp. & 76568 & 80D368 \\ 
		183 & SO & Southern co. & 18411 & 147C38 \\ 
		184 & LUV & Southwest airlines co. & 58683 & E866D2 \\ 
		185 & SWK & Stanley black \& decker inc. & 43350 & CE1002 \\ 
		186 & STT & State street corp. & 72726 & 5BC2F4 \\ 
		187 & TGNA & TEGNA inc. & 47941 & D6EAA3 \\ 
		188 & TXN & Texas instruments inc. & 15579 & 39BFF6 \\ 
		189 & TMK & Torchmark corp. & 62308 & E90C84 \\ 
		190 & TRV & The travelers companies inc. & 59459 & E206B0 \\ 
		191 & TBI & TrueBlue inc. & 83671 & 9D5D35 \\ 
		192 & TUP & Tupperware brands corp. & 83462 & 2B0AF4 \\ 
		193 & TYC & Tyco international plc & 45356 & 99333F \\ 
		194 & TSN & Tyson foods inc. & 77730 & AD1ACF \\ 
		195 & X & United states Steel corp. & 76644 & 4E2D94 \\ 
		196 & UNH & UnitedHealth group inc. & 92655 & 205AD5 \\ 
		197 & VIAV & Viavi solutions inc. & 79879 & E592F0 \\ 
		198 & GWW & W.W. grainger inc. & 52695 & 6EB9DA \\ 
		199 & WDR & Waddell \& reed financial inc. & 85931 & 2F24A5 \\ 
		200 & WBA & Walgreens boots alliance inc. & 19502 & FACF19 \\ 
		201 & DIS & Walt disney co. & 26403 & A18D3C \\ 
		202 & WAT & Waters corp. & 82651 & 1F9D90 \\ 
		203 & WBS & Webster financial corp. & 10932 & B5766D \\ 
		204 & WFC & Wells fargo \& co. & 38703 & E8846E \\ 
		205 & WERN & Werner enterprises inc. & 10397 & D78BF1 \\ 
		206 & WABC & Westamerica bancorp & 82107 & 622037 \\ 
		207 & WDC & Western digital corp. & 66384 & CE96E7 \\ 
		208 & WHR & Whirlpool corp. & 25419 & BDD12C \\ 
		209 & WFM & Whole foods market inc. & 77281 & 319E7D \\ 
		210 & XLNX & Xilinx inc. & 76201 & 373E85 \\ 
		\hline
		\caption{\footnotesize Final list of firms --- The table contains the information about the full list of firms: tickers, firm names, CRSP PERMNO code and RavenPack ID. Tickers and firm names are taken as of June, 2017. PERMNO and RavenPack ID columns are used to match firms and firm news data. \label{tab:list_firms}} 
	\end{longtable}
}

\end{document}